\documentclass[a4paper,english]{lipics-v2019}

\usepackage{microtype}

\usepackage{amsmath}
\usepackage{graphicx}
\usepackage[colorinlistoftodos]{todonotes}
\usepackage{stmaryrd}
\usepackage{algorithm} 
\usepackage{algorithmic}
\usepackage{verbatim}
\usepackage{balance}
\usepackage{comment}
\usepackage{cite}

\theoremstyle{plain}

\def\e#1{\emph{#1}}
\def\set#1{\mathord{\{#1\}}}

\def\prob{\mathrm{Pr}}
\def\P{\mathcal{P}}
\def\I{\mathcal{I}}

\def\K{\mathcal{K}}
\def\R{\mathcal{R}}
\def\U{\mathcal{U}}
\def\M{\mathcal{M}}
\def\D{\mathcal{D}}
\def\scs{\mathbf{S}}
\def\sig{\mathit{sig}}

\def\w{\mathbf{w}}

\newcommand{\eqdef}{\,\,{\overset{\mathrm{def}}{=}}\,\,}
\def\tids{\mathit{ids}}
\def\cells{\mathit{cells}}
\newcommand{\eat}[1]{}

\def\pudl{\mathbf{R}}

\def\tup#1{\mathbf{#1}}

\def\circs{{\circ}}
\def\tightcirc#1#2{\mathord{#1}{\circ}{#2}\,}
\def\str{^\star}
\def\vg{\mathsf{VG}}
\def\tg{\mathsf{TG}}

\def\limit{\longrightarrow}

\def\tuples{\mathrm{tuples}}
\def\const{\mathsf{Const}}

\def\dirty{J}

\def\intended{I}

\def\ra{\rightarrow}

\newenvironment{repeatresult}[2]
{\vskip0.5em\par\textsc{#1} #2.\em}
{\vskip1em}

\newenvironment{repproposition}[1]{\begin{repeatresult}{Proposition}{#1}}{\end{repeatresult}}
\newenvironment{reptheorem}[1]{\begin{repeatresult}{Theorem}{#1}}{\end{repeatresult}}

\title{A Formal Framework for Probabilistic Unclean Databases}

\author{Christopher De Sa}{Cornell University, Ithacan, NY, USA}{cdesa@cs.cornell.edu}{}{}

\author{Ihab F. Ilyas}{University of Waterloo, Waterloo, ON, Canada}{ilyas@uwaterloo.ca}{}{This work was supported by NSERC under a Discovery Grant.}

\author{Benny Kimelfeld}{Technion - Israel Institute of Technology, Haifa, Israel}{bennyk@cs.technion.ac.il}{}{This work was supported by the Israel Science Foundation (ISF) Grant 1295/15.}

\author{Christopher R{\'e}}{Stanford University, Stanford, CA, USA}{chrismre@cs.stanford.edu}{}{}

\author{Theodoros Rekatsinas}{University of Wisconsin - Madison, Madison, WI, USA}{thodrek@cs.wisc.edu}{}{This work was supported by the Wisconsin Alumni Association, Amazon under an ARA Award, and by NSF under grant IIS-1755676.}

\authorrunning{De Sa, Ilyas, Kimelfeld, R{\'e}, and Rekatsinas}

\Copyright{Christopher De Sa and Ihab F. Ilyas and Benny Kimelfeld and Christopher R{\'e} and Theodoros Rekatsinas}



\category{}

\relatedversion{}

\supplement{https://arxiv.org/abs/1801.06750}

\funding{}

\acknowledgements{}

\nolinenumbers

\EventEditors{Pablo Barcelo and Marco Calautti}
\EventNoEds{2}
\EventLongTitle{22nd International Conference on Database Theory (ICDT 2019)}
\EventShortTitle{ICDT 2019}
\EventAcronym{ICDT}
\EventYear{2019}
\EventDate{March 26--28, 2019}
\EventLocation{Lisbon, Portugal}
\EventLogo{}
\SeriesVolume{127}
\ArticleNo{3} 

\begin{document}

\maketitle

\begin{abstract}
  Most theoretical frameworks that focus on data errors and
  inconsistencies follow logic-based reasoning. Yet, practical data
  cleaning tools need to incorporate statistical reasoning to be
  effective in real-world data cleaning tasks. Motivated by empirical
  successes, we propose a formal framework for unclean databases,
  where two types of statistical knowledge are incorporated: The first
  represents a belief of how intended (clean) data is generated, and
  the second represents a belief of how noise is introduced in the
  actual observed database. To capture this noisy channel model, we
  introduce the concept of a Probabilistic Unclean Database (PUD), a
  triple that consists of a probabilistic database that we call the
  \emph{intention}, a probabilistic data transformator that we call
  the \emph{realization} and captures how noise is introduced, and an
  observed unclean database that we call the \emph{observation}. We
  define three computational problems in the PUD framework: cleaning
  (infer the most probable intended database, given a PUD),
  probabilistic query answering (compute the probability of an answer
  tuple over the unclean observed database), and learning (estimate
  the most likely intention and realization models of a PUD, given
  examples as training data). We illustrate the PUD framework on
  concrete representations of the intention and realization, show that
  they generalize traditional concepts of repairs such as cardinality
  and value repairs, draw connections to consistent query answering,
  and prove tractability results. We further show that parameters can
  be learned in some practical instantiations, and in fact, prove that
  under certain conditions we can learn a PUD directly from a single
  dirty database without any need for clean examples.
\end{abstract}

\section{Introduction}
\label{sec:intro}

Managing errors and inconsistency in databases is traditionally viewed
as a challenge of a logical nature.  It is typical that errors in a
database are defined with respect to \e{integrity constraints} that
capture normative aspects of downstream applications. The aim of
integrity constraints is to guarantee the consistency of data used by
these applications.  Typically, an unclean database is defined as a
database $\dirty$ that violates the underlying set of integrity
constraints.  In turn, a \e{repair} of a database $\dirty$ is a clean
database $\intended$ wherein all integrity constraints hold, and is
obtained from $\dirty$ by a set of operations (e.g., deletions of
tuples or updates of tuple values) that feature some form of
non-redundancy~\cite{DBLP:conf/pods/ArenasBC99,DBLP:conf/icdt/AfratiK09}.

Various computational problems around unclean databases have been
investigated in prior
work~\cite{DBLP:conf/pods/ArenasBC99,Lenzerini:2002:DIT:543613.543644,DBLP:conf/icdt/LopatenkoB07,
  Libkin:2014:IDW:2594538.2594561}.  Past theoretical research has
established fundamental results that concentrate on tractability
boundaries for repair-checking and consistent query
answering~\cite{DBLP:conf/icdt/AfratiK09,Fagin:2015:DCP:2745754.2745762,DBLP:journals/tods/KoutrisW17}.
In their majority, these works adopt a deterministic interpretation of
data repairs and cast all repairs equally likely.  These theoretical
developments have inspired practical tools that aim to automate data
cleaning~\cite{DBLP:conf/icde/BohannonFGJK07, YakoutENOI11,
  DBLP:conf/icde/xu13, DBLP:conf/sigmod/WangT14,Ilyas2016EffectiveDC}.
The majority of proposed methods assume as input a set of integrity
constraints and use those to identify possible repairs via
search-based procedures.  To prioritize across possible repairs during
search, the proposed methods rely on the notion of {\em
  minimality}~\cite{chomicki2005minimal,
  DBLP:conf/icdt/KolahiL09,DBLP:conf/pods/LivshitsKR18}.  Informally,
minimality states that given two candidate sets of repairs, the one
with fewer changes with respect to the original database is
preferable.  The use of minimality as an operational principle to find
data repairs is a practical artifact that is used to limit the search
space.  These approaches to data cleaning suffer from two major
drawbacks: First, they do not permit concrete statements about the
``likelihood'' of possible repairs.  Consequently, they categorize
query answers to a limited set of validity labels (e.g., certain,
possible, and impossible); these labels might be unsuitable for
downstream applications.  Second, combinatorial principles such as
minimality, while desired, do not entail the richness of the arguments
and evidences (e.g., statistical features of data) that are needed to
reason about and generate correct repairs.

Effective data cleaning needs to incorporate statistical reasoning.
Our recent work on HoloClean~\cite{holoclean} casts data repairing as
a statistical learning and inference problem and reasons about a
\e{most probable} repair instead of a \e{minimal} repair.  Our study
shows that HoloClean obtains more accurate data cleaning results than
competing minimality-based data cleaning tools for a diverse array of
real-world data cleaning scenarios~\cite{holoclean}.  HoloClean uses
training data to learn a probabilistic model for how clean data is
generated and how data errors are injected.  HoloClean's model follows
the {\em noisy channel model}~\cite{Jurafsky:2009:SLP:1214993}, the
de-facto probabilistic framework used in natural language tasks, such
as spell checking and speech recognition, to reason about noisy data.
To the best of our knowledge, existing theoretical frameworks for data
cleaning do not capture this type of probabilistic reasoning.

\subparagraph*{Goals}
We aim to establish a formal framework for \e{probabilistic unclean
  databases} (PUD) that adopts a statistical view of database
cleaning.  We do so by following the aforementioned noisy channel
paradigm of HoloClean.  Within the PUD framework, we formalize
fundamental computational problems: \e{cleaning}, \e{query answering},
and \e{learning}.  With that, we aim to draw connections between
theoretical database research and important aspects of practical
systems.  In particular, our goal is to open the way for analyses and
algorithms with theoretical guarantees for such systems.  We argue
that our framework is basic enough to allow for nontrivial theoretical
advances, as illustrated by our preliminary results that \e{(a)} draw
connections to traditional deterministic concepts, \e{and (b)} devise
algorithms for special cases.

\subparagraph*{Probabilistic unclean databases} We view an unclean
database as if a clean database $I$ had been ``distorted'' via a noisy
channel into a dirty database $J$; we aim to establish a model of this
channel.  Given the observed unclean database $J$, we seek the true
database $I$ from which $J$ is produced.  This model adopts Bayesian
inference: out of all possible $I$, we seek the one for which the
probability, given $J$, is highest.  Following Bayes' rule, our
objective is to find $\arg\max_{I}\prob(I)\cdot\prob(J|I)$.  This
objective decomposes in two parts: (1) the \e{prior} model for a clean
database captured by $\prob(I)$, and (2) the \e{channel} or \e{error}
model characterized by $\prob(J|I)$.  To capture that, we define a
\e{Probabilistic Unclean Database} (\e{PUD}) as a triple
$(\I,\R,\dirty^\star)$ where: (1) $\I$, referred to as the
\e{intention model}, is a distribution that produces intended clean
databases; (2) $\R$, referred to as the \e{realization model}, is a
function that maps each clean database $I$ to a distribution $\R_{I}$
that defines how noise is introduced into $I$; and (3) $J\str$ is an
observed unclean database. The distribution $\I$ defines the prior
$\prob(I)$ over clean databases, while the distribution $\R_{I}$
defines the aforementioned noisy channel $\prob(J|I)$.

\subparagraph*{Computational problems} We define and study three
computational problems in the PUD framework: (1) \e{data cleaning},
where given a PUD $(\I,\R,\dirty)$, we seek to compute a database
$\intended$ that maximizes the probability
$\I(\intended)\times\R_{\intended}(\dirty)$; (2) \e{probabilistic
  query answering}, that is, the problem of evaluating a query $Q$
over a PUD following the traditional possible tuple
semantics~\cite{DBLP:conf/vldb/DalviS04,Suciu:2011:PD:2031527}; and
(3) \e{learning} a PUD, where we consider \e{parametric}
representations $\I_{\Xi}$ and $\R_{\Theta}$ of the intention and
realization models, and seek to estimate the parameter vectors $\Xi^*$
and $\Theta^*$ that maximize the likelihood of training data.

\subparagraph*{Preliminary analysis} PUDs allow for different
instantiations of the intention and realization models.  To establish
preliminary complexity and convergence results, we focus on specific
instantiations of the intention and realization models.  We study
intention models that can describe the distribution of tuple values as
well as both soft and hard integrity constraints.  We also focus on
simple noise models. We study (1) realizations that introduce new
tuples, hence, the clean database is a subset of the observed unclean
database, and (2) realizations that update table cells, hence, the
clean database is obtained via value repairs over the observed unclean
database.

We present PUD instantiations for which solving the data cleaning
problem has polynomial-time complexity. For instance, we show that in
the presence of only one key constraint, soft or hard, data cleaning
in PUDs can be solved in polynomial time. This result extends results
for deterministic repairs that focus on hard integrity constraints to
\e{weak} (soft) key constraints (e.g., two people are unlikely to, but
might, have the same first and last name). Here, the most probable
repair under the PUD framework may violate weak key constraints.  We
also draw connections between data cleaning in the PUD framework and
\e{minimal repairs}. We identify conditions under which data cleaning
in the PUD framework is equivalent to \e{cardinality
  repairs}~\cite{DBLP:conf/icdt/LopatenkoB07} and \e{optimal
  V-repairs}~\cite{DBLP:conf/icdt/KolahiL09}.  For PUD learning, we
consider both \e{supervised} and \e{unsupervised} learning.  In the
former case, we are given intension-realization pairs, and in the
former, we are given only realizations (i.e., dirty databases). Our
results discuss convexity and gradient computation for the
optimization problem underlying the learning problem.

Our PUD model can be viewed as a generalization of the approach of
Gribkoff et al.~\cite{GVSBUDA14}, who view the dirty database as a
tuple-independent probabilistic
database~\cite{DBLP:conf/vldb/DalviS04}, and seek the \e{most-probable
  database} that satisfies a set of underlying integrity constraints
(e.g., functional dependencies). In contrast, our modeling allows for
arbitrary distributions over the intention, including ones with
\e{weak} constraints that we discuss later on. Interestingly, our PUD
model goes in the reverse direction of the \e{operational} approach of
Calautti et al.~\cite{DBLP:conf/pods/CalauttiLP18}, who view the dirty
database as a deterministic object and its \e{cleaning} (rather than
the \e{error}) as a probabilistic process (namely a Markov chain of
repairing operations).

\subparagraph*{Vision} This paper falls within the bigger vision of
bridging database theory with learning theory as outlined in a recent
position article~\cite{DBLP:journals/sigmod/AbiteboulABBCDH16}. We aim
to draw connections between the rich theory on inconsistency
management by the database community, and fundamentals of statistical
learning theory with emphasis on \e{structured
  prediction}~\cite{Bakir:2007:PSD:1296180}. Structured prediction
typically focuses on problems where, given a collection of
observations, one seeks to predict the most likely assignment of
values to structured objects. In most practical structured prediction
problems, structure is encoded via logic-based
constraints~\cite{Globerson:2015:HIS:3045118.3045350} in a way similar
to how consistency is enforced in data cleaning. It is our hope that
this paper will commence a line of work towards theoretical
developments that take the benefit of both worlds, and will lead to
new techniques that are both practical and rooted in strong
foundations.

\subparagraph*{Organization} We begin with preliminary definitions in Section~\ref{sec:prelims}. In Section~\ref{sec:puds} we present the concept of PUDs. We present the three fundamental computational problems in Section~\ref{sec:problems}, and describe preliminary results in Sections~\ref{sec:mlipqa} and~\ref{sec:learning}. We conclude with a discussion in Section~\ref{sec:conclusion}. For space limitations, all proofs are in the Appendix of our paper.

\section{Preliminaries}

\label{sec:prelims}
We first introduce concepts, definitions and notation that we need
throughout the paper. 
 

\subparagraph*{Schemas and databases} A \e{relation signature} is a
sequence $\alpha=(A_1,\dots,A_k)$ of distinct \e{attributes} $A_i$,
where $k$ is the \e{arity} of $\alpha$. A (\e{relational}) \e{schema}
$\scs$ has a finite set of \e{relation symbols}, and it associates
each relation symbol $R$ with a signature that we denote by
$\sig_\scs(R)$, or just $\sig(R)$ if $\scs$ is clear from the context.
We assume an infinite domain $\const$ of \e{constants}. Let $\scs$ be
a schema, and let $R$ be a relation symbol of $\scs$. A \e{tuple} $t$
over $R$ is a sequence $(c_1,\dots,c_k)$ of constants, where $k$ is
the arity of $\sig(R)$. If $t=(c_1,\dots,c_k)$ is a tuple over $R$ and
$\sig(R)=(A_1,\dots,A_k)$, then we refer to the value $c_j$ as $t.A_j$
(where $j=1,\dots,k$). We denote by $\tuples(R)$ the set of all tuples
over $R$.

In our databases, tuples have unique record identifiers. Formally, a
table $r$ over $R$ is associated with a finite set $\tids(r)$ of
identifiers, and it maps each identifier $i$ to a tuple $r[i]$ over
$R$. A \e{database} $I$ over $\scs$ consists of a table $R^I$ over
each relation symbol $R$ of $\scs$, such that no two occurrences of
tuples have the same identifier; that is, if $R_1$ and $R_2$ are
distinct relation symbols in $\scs$, then $\tids(R_1^I)$ and
$\tids(R_2^I)$ are disjoint sets. We denote by $\tids(I)$ the union of
the sets $\tids(R^I)$ over all relation symbols $R$ of $\scs$. If
$i\in\tids(R^I)$, then we may refer to the tuple $R^I[i]$ simply as
$I[i]$.

A \e{cell} of a database $I$ is a pair $(i,A)$, where $i\in\tids(R^I)$
for a relation symbol $R$, and $A$ is an attribute inside
$\sig(R)$. We denote the cell $(i,A)$ also by $i.A$, and we denote by
$\cells(I)$ the set of all cells of $I$.

Let $I$ and $J$ be databases over the same schema $\scs$. We say that
$I$ is a \e{subset} of $J$ if $I$ can be obtained from $J$ by deleting
tuples, that is, $\tids(R^I)\subseteq\tids(R^J)$ for all relation
symbols $R$ of $\scs$ (hence, $\tids(I)\subseteq\tids(J)$) and
$I[i]=J[i]$ for all $i\in\tids(I)$. We say that $I$ is an \e{update}
of $J$ if $I$ can be obtained from $J$ by changing attribute values,
that is, $\tids(R^J)=\tids(R^I)$ for all relation symbols $R$ of
$\scs$.

A \e{query} $Q$ over a schema $\scs$ is associated with fixed arity,
and it maps every database $D$ over $\scs$ into a finite set $Q(D)$ of
tuples of constants over the fixed arity.


\subparagraph*{Integrity constraints}
Various types of logical conditions are used for declaring integrity
constraints, including \e{Functional Dependencies} (FDs),
\e{conditional FDs}~\cite{DBLP:conf/icde/BohannonFGJK07}, \e{Denial
  Constraints} (DCs)~\cite{DBLP:journals/jiis/GaasterlandGM92},
referential constraints~\cite{Date:1981:RI:1286831.1286832}, and so
on. In this paper, by \e{integrity constraint} over a schema $\scs$ we
refer to a general expression $\varphi$ of the form
$\forall{x_1,\dots,x_m}[\gamma(x_1,\dots,x_m)]$, where
$\gamma(x_1,\dots,x_m)$ is a safe expression in Tuple Relational
Calculus (TRC) over $\scs$.  For example, an FD $R:A\ra B$ is
expressed here as the integrity constraint
\[\forall x,y \left[(x\in R y\in R) \ra (x.A=y.A \ra x.B=y.B)\right]\,.\]
A \e{violation} of
$\varphi=\forall{x_1,\dots,x_m}[\gamma(x_1,\dots,x_m)]$ in the
database $I$ is a sequence $i_1,\dots,i_m$ of tuple identifiers in
$\tids(I)$ such that $I$ violates $\gamma(I[i_1],\dots,I[i_m])$, and
we denote by $V(\varphi,I)$ the set of violations of $\varphi$ in $I$.
We say that $I$ \e{satisfies} $\varphi$ if $I$ has no violations of
$\varphi$, that is, $V(\varphi,I)$ is empty. Finally, $I$
\e{satisfies} a set $\Phi$ of integrity constraints if $I$ satisfies
every integrity constraint $\varphi$ in $\Phi$.

\subparagraph*{Minimum repairs}
Traditionally, database \e{repairs} are defined over inconsistent
databases, where \e{inconsistencies} are manifested as violations of
integrity constraints. A repair is a consistent database that is
obtained from the inconsistent one by applying a \e{minimal} change,
and we recall two types of repairs: \e{subset} (obtained by deleting
tuples) and \e{update} (obtained by changing values). Moreover, the
repairing operations may be weighted by tuple weights (in the first
case) and cell weights (in the second case). 

Formally, let $\scs$ be a schema, $\Phi$ a set of integrity
constraints over $\scs$, and $J$ a database that does not necessarily
satisfy $\Phi$. A \e{consistent subset} (resp., \e{consistent update})
of $J$ is a subset (resp., update) $I$ of $J$ such that $I$ satisfies
$\Phi$. A \e{minimum subset repair} of $J$ w.r.t.~a weight function
$w:\tids(J)\rightarrow[0,\infty)$ is a consistent subset $I$ of $J$
that minimizes the sum $\sum_{i\in\tids(J)\setminus\tids(I)}w(i)$.  As
a special case, a \e{cardinality repair} of $J$ is a minimum subset
repair w.r.t.~a constant weight (e.g., $w(i)=1$), that is, a
consistent subset with a maximal number of tuples.  A \e{minimum
  update repair} of $J$ w.r.t.~a weight function
$w:\cells(J)\times\const\rightarrow[0,\infty)$ is a consistent update
$I$ of $J$ that minimizes the sum
$\sum_{i.A\in\cells(I)}w(i.A,I[i].A)$.

\subparagraph*{Probabilistic databases} A \e{probabilisitic database}
is a probability distribution over ordinary databases. As a
representation system, our model is a generalization of the
\e{Tuple-Independent probabilistic Database} (TID) wherein each tuple
might either exist (with an associated probability) or
not~\cite{DBLP:conf/vldb/DalviS04,Suciu:2011:PD:2031527}. In our
model, each tuple comes from a general probability distribution over
tuples (where inexistence is one of the options). This allows us to
incorporate beliefs about the likelihood of tuples and cell values.

We now give the formal definition. Let $\scs$ be a schema. A
\e{generalized TID} is a database $\K$ that is defined similarly to an
ordinary database over $\scs$, except that instead of a tuple, the
entry $R^{\K}[i]$ is a discrete probability distribution over the set
$\tuples(R)\cup\set{\bot}$, where the special value $\bot$ denotes
that no tuple is generated. Hence, for every tuple $t$ over $R$, the
probability that $R^{\K}[i]$ produces $t$ is given by $R^{\K}[i](t)$,
or just $\K[i](t)$; moreover, the number $\K[i](\bot)$ is the
probability that no tuple is generated for the identifier
$i$. Therefore, $\K$ defines a probability distribution over databases
$I$ over $\scs$ such that $\tids(I)\subseteq\tids(\K)$ and the
probability $\K(I)$ of a database $I$ is defined as follows:
\[\K(I)\eqdef \prod_{i\in\tids(I)}\K[i](I[i])\,\,\times\,\,
\prod_{i\in\tids(\K)\setminus\tids(I)}\hspace{-1.5em}\K[i](\bot)\]

We incorporate weak integrity constraints by adopting the standard
concept of \e{parametric factors} (or \e{parfactors} for short), which
has been used in the \e{soft keys} of Jha et
al.~\cite{DBLP:conf/pods/JhaRS08} and the \e{PrDB} model of Sen et
al.~\cite{DBLP:journals/vldb/SenDG09}, and which can be viewed as a
special case of the \e{Markov Logic Network}
(MLN)~\cite{Richardson:2006:MLN:1113907.1113910}. Under this concept,
each constraint $\varphi$ is associated with a weight $w(\varphi)>0$
and each violation of $\varphi$ contributes a factor of
$\exp(-w(\varphi))$ to the probability of a random database $I$.
Formally, a \e{parfactor database} over a schema $\scs$ is a triple
$\D=(\K,\Phi,w)$, where $\K$ is a generalized TID, $\Phi$ is a finite
set of integrity constraints, both over $\scs$, and
$w:\Phi\ra(0,\infty)$ is a weight function over $\Phi$.\eat{Compared
  to an MLN, a parfactor database allows to separate the \e{data}
  (i.e., $\K$) from the \e{constraints} (i.e., $\Phi$) for the sake of
  complexity analysis. For example, to model a TID, an MLN requires a
  constraint for each tuple, and a parfactor database requires no
  constraints at all.} The probability $\D(I)$ of a database $I$ is
defined as follows.
\[ \D(I) \eqdef \frac{1}{Z}\,\,\times\,\, \K(I)\,\,\times\,\,
\exp\left(
-\sum_{\varphi\in\Phi}
w(\varphi)\times |V(\varphi,I)|\right)
\]
Recall that $V(\varphi,I)$ the set of violations of $\varphi$ in $I$.
The number $Z$ is a \e{normalization factor} (also called the
\e{partition
  function}) that normalizes the sum of probabilities to one:
\[Z\eqdef 
\sum_{I}\K(I)\,\times\,
\exp\left(
-\sum_{\varphi\in\Phi}
w(\varphi)\times |V(\varphi,I)|\right)
\]
Observe that the above sum is over a countable domain, since we assume
that every $R^{\K}[i]$ is discrete (hence, there are countably many
random databases $I$). Since we normalize the probability, it is not
really necessarily for $\K$ to be normalized, as $\D$ would be a
probability distribution even if $\K$ is \e{not} normalized. In fact,
in our analysis, we will \e{not} make the assumption that $\K$ is
normalized.

\begin{table}
  \caption{\label{table:symb}Main symbols used in the framework.}
\scriptsize
\renewcommand{\arraystretch}{1.4}
\begin{tabular}{c|l}\hline
  $\scs$ & A schema. \\
  $\U$ & A PUD $(\I,\R,J\str)$.\\
  $\I$ & An intention model (probabilistic database).\\
  $\R$ & A realization model, maps every $I$ into a probabilistic database $\R_I$. \\
  $J\str$ & An observed unclean database.\\
  $\R\circs\I$ & Distribution  over pairs $(I,J)$ given by $\tightcirc{\R}{\I}(I,J)=\I(I)\cdot\R_{I}(J)$.\\
  $\U\str$ &  A probabilistic database given by $\U\str(I')=\prob_{(I,J)\sim\R\circ\I}(I=I'\mid J=J\str)$.\\
$(\D,\tau,J\str)$ & A parfactor/subset PUD.\\
$(\D,\kappa,J\str)$ & A parfactor/update PUD.\\
  $\D$ & A parfactor database $(\K,\Phi,w)$ with $w:\Phi\ra(0,\infty)$.\\
  $\K$ & A generalized tuple-independent database (generalized TID).\\
  $\Phi$ & A set of integrity constraints $\varphi$.\\
  $\tau$ & Maps $i\in\tids(R^{\D})$ to a discrete distribution $\tau[i]$ over $\tuples(R)\cup\set{\bot}$.\\
  $\kappa$ & Maps $(i,t)\in\tids(R^{\D})\times\tuples(R)$ to a discrete distribution $\kappa[i,t]$ over $\tuples(R)$.\\
\hline
\end{tabular}
\end{table}

\section{Probabilistic Unclean Databases}\label{sec:puds}

We introduce the Probabilistic Unclean Database (PUD) framework and
describe examples of PUD instantiations that correspond to data
cleaning applications in the HoloClean system~\cite{holoclean}.  In
our framework, a PUD consists of three components following a
noisy-channel model: \e{(1)} an \e{intention} model for generating
clean databases, \e{(2)} a noisy \e{realization} model that can
distort the intended clean database, \e{and (3)} an observed unclean
database. The formal definition follows.

\begin{definition} Let $\scs$ be a schema.  A PUD (over $\scs$) is a
  triple $\U =(\I,\R,J\str)$ where:
\begin{enumerate}
\item $\I$ is a probabilistic database, referred to as the
  \e{intention} model;
\item $\R$, referred to as the \e{realization} model, is a function
  that maps each database $I$ to a probabilistic database $\R_I$;
\item $J\str$ is a database referred to as the \e{observed} or
  \e{unclean} database.
\end{enumerate}
\end{definition}

\begin{figure}[t]
  \centering
  \includegraphics[width=1\textwidth]{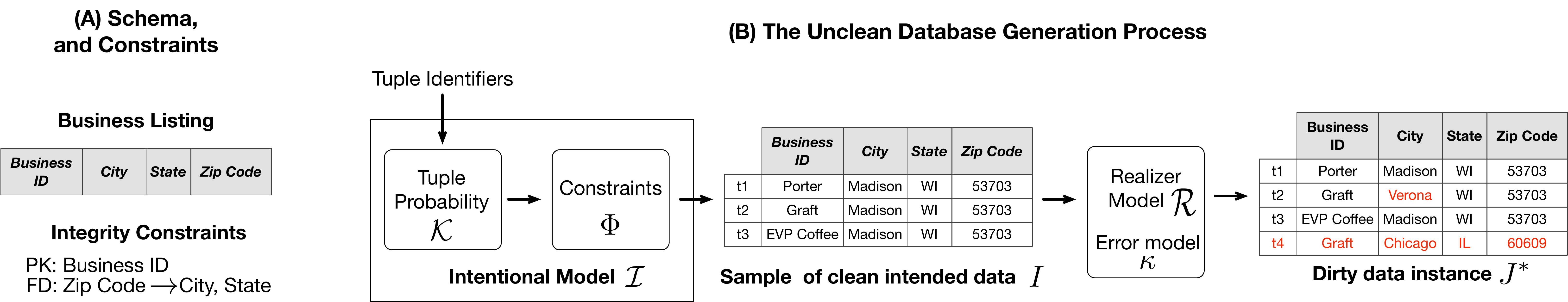}
  \caption{Overview of the PUD framework.}
  \label{fig:pudex}
\end{figure}

\begin{example}\label{ex:puds}
  Figure~\ref{fig:pudex} illustrates a high-level example of the PUD
  framework.  We use a running example from business
  listings. Figure~\ref{fig:pudex}(A) depicts the schema $\scs$ of the
  example. The constraints include a primary key and a functional
  dependency. Figure~\ref{fig:pudex}(B) depicts the unclean database
  generation process. Intention $\I$ outputs a valid database $I$ with
  three tuples. The realizer $\R$ takes as input this database $I$,
  injects the new tuple $t4$ and updates the \texttt{City} value of
  tuple $t2$ from ``Madison'' to ``Verona.'' \qed
\end{example}

A PUD $\U = (\I,\R,J\str)$ defines a probability distribution, denoted
$\R\circs\I$, over pairs $(I,J)$. Conditioning on $J=J\str$, the PUD
$\U$ also defines a probability distribution, denoted $\U\str$, over
intentions $I$ (i.e., a probabilistic database).  In the generative
process of $\R\circs\I$, we sample the intention $I$ from $\I$, and
then we sample $J$ from the realization $\R_I$. Hence, the probability
of $(I,J)$ is given by 
\[\tightcirc{\R}{\I}(I,J)\eqdef\I(I)\cdot\R_{I}(J)\,.\]
In the probabilistic database $\U\str$, the probability of each
candidate intention $I'$ is given by
\[\U\str(I')\,\,\eqdef\,\,\prob_{(I,J)\sim\R\circ\I}(I=I'\mid J=J\str)
\,\,=\,\,\frac{\tightcirc{\R}{\I}(I',J\str)}{\sum_{I}\tightcirc{\R}{\I}(I,J\str)}
\]
that is, the probability conditioned on the random $J$ being $J\str$.
For this distribution to be well defined, we require $J\str$ to have
a nonzero probability; that is, there exists $I$ such that
$\tightcirc{\R}{\I}(I,J\str)>0$. Table~\ref{table:symb} lists the main symbols in the framework, along with their meaning.




\subsection{Example Instantiations of PUDs}\label{sec:applications}
\label{sec:pudexamples}
Our definition of a PUD is abstract, and not associated with any
specific representation model.  We now present concrete instantiations
of PUD representations. These instantiations are probabilistic
generalizations of the (deterministic) concepts of \e{subset
  repairs}~\cite{DBLP:conf/icdt/LopatenkoB07,DBLP:conf/icdt/AfratiK09}
and \e{update repair}~\cite{DBLP:conf/icdt/KolahiL09,DBLP:conf/pods/LivshitsKR18},
respectively. More precisely, in both instantiations, the PUD $\U =
(\I,\R,J\str)$ is such that $\I$ is represented as a parfactor
database $\D$ (as defined in Section~\ref{sec:prelims}) and $J\str$ is
an ordinary database (as expected); the two differ in the
representation of the realization model $\R$. In the first
instantiation, $\R$ is allowed to introduce new random tuples (hence,
the intended database is a \e{subset} of the unclean one) and in the
second, $\R$ is allowed to randomly change tuples (hence, the intended
database is an \e{update} of the unclean one).  Formally, let $\scs$
be a schema.
\begin{itemize}
\item A \e{parfactor/subset} PUD is a triple $(\D,\tau,J\str)$ where
  $\D$ is a parfactor database, $\tau$ maps every identifier
  $i\in\tids(R^{\D})$, where $R\in\scs$, to a discrete distribution
  $\tau[i]$ over $\tuples(R)\cup\set{\bot}$, and $J\str$ is an
  ordinary database. As usual, $\bot$ means that no tuple is
  generated.
\item A \e{parfactor/update} PUD is a triple $(\D,\kappa,J\str)$ where
  $\D$ is a parfactor database, $\kappa$ maps every identifier
  $i\in\tids(R^{\D})$ and tuple $t\in\tuples(R)$, where $R\in\scs$, to
  a discrete distribution $\kappa[i,t]$ over $\tuples(R)$, and $J\str$
  is an ordinary database.
\end{itemize}

In a parfactor/subset PUD $\U=(\D,\tau,J\str)$, the probability
$\tightcirc{\R}{\I}(I,J)$ is then defined as follows. If $I$ is not a
subset of $J$, then $\tightcirc{\R}{\I}(I,J)=0$; otherwise:
\begin{align*}
\tightcirc{\R}{\I}(I,J)\eqdef
\D(I)\times
\prod_{\substack{i\in\tids(J)\setminus\\\tids(I)}}
\hskip-1em\tau[i](J[i])
\times 
\prod_{\substack{i\in\tids(\D)\setminus\\\tids(J)}}
\hskip-1em\tau[i](\bot) 
\end{align*}
That is, $\tightcirc{\R}{\I}(I,J)$ is the probability of $I$ (i.e.,
$\D(I)$), multiplied by the probability that each new tuple of $J$ is
produced by $\tau$ (i.e., $\tau[i](J[i])$), multiplied by the
probability that each tuple identifier $i$ missing in $J$ is indeed
not produced (i.e., $\tau[i](\bot)$).

In a parfactor/update PUD $\U=(\D,\kappa,J\str)$, the probability
$\tightcirc{\R}{\I}(I,J)$ is then defined as follows. If $I$ is not an
update of $J$, then $\tightcirc{\R}{\I}(I,J)=0$; otherwise:
\begin{align*}
\tightcirc{\R}{\I}(I,J)\eqdef
\D(I)\times 
\prod_{\substack{i\in\tids(I)}} 
\hskip-0.5em
\kappa[i,I[i]](J[i])
\end{align*}
That is, $\tightcirc{\R}{\I}(I,J)$ is the probability of $I$ (i.e.,
$\D(I)$), multiplied by the probability that $\kappa$ changes each
tuple $I[i]$ to $J[i]$ (i.e., $\kappa[i,I[i]](J[i])$). 

\begin{figure}[t]
  \centering
  \includegraphics[width=1\textwidth]{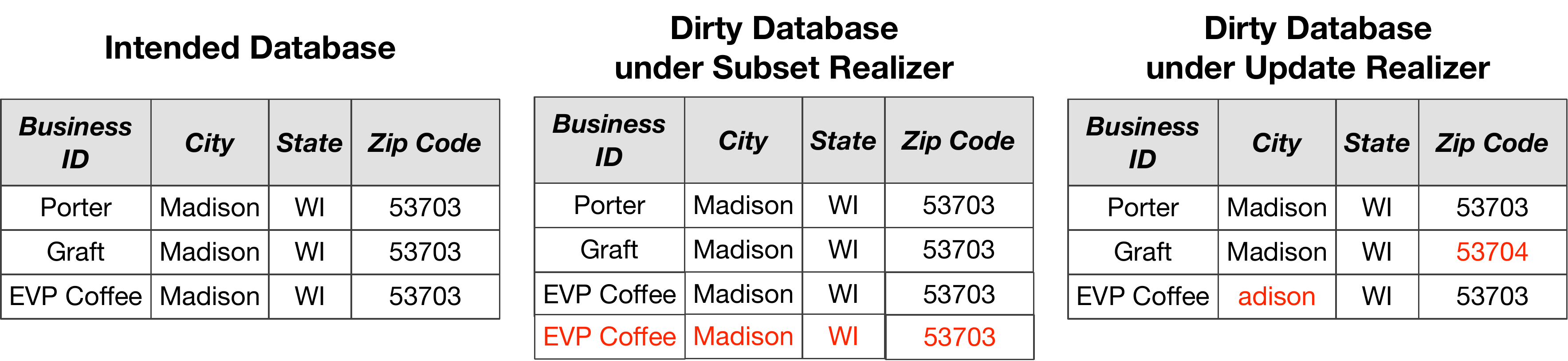}
  \caption{Examples of a subset realizer and an update realizer.}
  \label{fig:realex}
\end{figure}

\begin{example}
  Figure~\ref{fig:realex} shows the intended database from
  Example~\ref{ex:puds} and two unclean versions obtained by a subset
  realizer and an update realizer. The subset realizer introduces a
  duplicate, while the update realizer introduces two typos. These
  correspond to two types of common errors in relational data. Our PUD
  framework can naturally model such cases.\qed
\end{example}

In Section~\ref{sec:mlipqa}, we discuss connections between these PUD instantiations and the deterministic models. Finally, in Section~\ref{sec:learning}, we provide more concrete cases of PUD instantiations.

\section{Computational Problems}
\label{sec:problems}
We define three computational problems over PUDs that are motivated by
the need to clean and query unclean data, and learn the intention and
realization models from observed data.

\subparagraph*{Data Cleaning} 
Given a PUD $(\I,\R,J\str)$, we wish to compute a Most Likely
Intention (MLI) database $I$, given the observed unclean database
$J\str$. We refer to this problem as \e{data cleaning} in PUDs.

\begin{definition}[Cleaning]
  Let $\scs$ be a schema and $\pudl$ a representation system for PUDs.
  The problem \e{$(\scs,\pudl)$-cleaning} is that of computing an
  \e{MLI} of a given PUD $\U =(\I,\R,J\str)$, that is, computing a
  database $I$ such that the probability $\U\str(I)$ is maximal (or,
  equivalently, the probability $\tightcirc{\R}{\I}(I,J\str)$ is maximal).
\end{definition}

\subparagraph*{Probabilistic query answering}
A PUD defines a probabilistic database---a probability space over the
intensions $I$.  The problem of \e{Probabilistic Query Answering}
(\e{PQA}) is that of evaluating a query over this probabilistic
database. We adopt the standard semantics of query evaluation over
probabilistic
databases~\cite{DBLP:conf/vldb/DalviS04,Suciu:2011:PD:2031527}, where
the confidence in an answer tuple is its marginal probability.

\begin{definition}[PQA]
  Let $\scs$ be a schema, $Q$ a query over $\scs$, and $\pudl$ a
  representation system for PUDs. The problem \e{$(\scs,Q,\pudl)$-PQA}
  is the following.  Given a PUD $\U$ and a tuple $\tup a$, compute
  the \e{confidence} of $\tup a$, that is, the probability
  $\prob_{I\sim\U\str}(\tup a\in Q(I))$.
\end{definition}

For now, we assume that both $\I$ and $\R$ are fully specified. We
next define the problem of learning models $\I$ and $\R$ using
training (potentially labeled) data.

\subparagraph*{PUD learning} For a PUD $(\I,\R,J\str)$, the models $\I$ and
$\R$ are typically represented using numeric \e{parameters}. For
example, the parameters of a parfactor/subset PUD $(\D,\tau,J\str)$
are those needed to represent $\D$ (e.g., the weights of the
constraints), and the parameters that define the distributions over
the tuples in both $\D$ and $\tau$. By a \e{parametric intention} we
refer to an intension model $\I_{\Xi}$ with a vector $\Xi$
of uninitialized parameters, and by $\I_{\Xi/\mathbf{c}}$ we denote the
actual intention model where $\Xi$ is assigned the values in the
vector $\mathbf{c}$. Similarly, by a \e{parametric realization} we
refer to a realization model $\R_{\Theta}$ with a vector $\Theta$ of
uninitialized parameters, and by $\R_{\Theta/\mathbf{d}}$ we denote
the actual realization model where $\Theta$ is set to $\mathbf{d}$.

Following the concept of \e{maximum likelihood estimation}, the goal
in learning is to find the parameters that best explain (i.e.,
maximize the probability) of the training examples. In the
\e{supervised} variant, we are given examples of both unclean
databases and their clean versions; in the \e{unsupervised} variant,
we are given only unclean databases.

\begin{definition}[Learning]\label{def:learning}
  Let $\scs$ be a schema, and $\pudl$ a representation system for
  parametric intensions and realizations. In the following problems we
  are given, as part of the input, the parametric intention and
  realization models $\I_{\Xi}$ and $\R_{\Theta}$, respectively.
\begin{itemize}
\item In the \e{supervised $(\scs,\pudl)$-learning} problem, we are
  also given a collection $(I_j,J_j)_{j=1}^n$ of database pairs
  (intention-realization examples), and the goal is to find parameter
  values $\mathbf{c}$ and $\mathbf{d}$ that maximize
  $\prod_{j=1}^n\tightcirc{\R}{\I}(I_j,J_j)$ for
  $\I=\I_{\Xi/\mathbf{c}}$ and $\R=\R_{\Theta/\mathbf{d}}$.
\item In the \e{unsupervised $(\scs,\pudl)$-learning} problem, we are
  also given a collection $(J_j)_{j=1}^n$ of databases (realization
  examples), and the goal is to find parameter values $\mathbf{c}$ and
  $\mathbf{d}$ that maximize $\prod_{j=1}^n\tightcirc{\R}{\I}(J_j)$
  for $\I=\I_{\Xi/\mathbf{c}}$ and $\R=\R_{\Theta/\mathbf{d}}$, where
  $\tightcirc{\R}{\I}(J_j)$ is the marginal probability of $J_j$, that
  is, $\sum_I\tightcirc{\R}{\I}(I,J_j)$.
\end{itemize}
\end{definition}

Note that the summation in the unsupervised variant is over the sample
space of the intention model $\I$.  While the reader might be
concerned about the source of many examples $(I_j,J_j)$ and $J_j$ in
the phrasing of the learning problems, it is oftentimes the case that
a single large example $(I,J)$ (or just $J$ in the unsupervised
variant) can be decomposed into many smaller examples. This depends on
the independence assumptions in the parametric models $\I_{\Xi}$ and
$\R_{\Theta}$ as we discuss in Section~\ref{sec:learningsetup}. In the
next sections, we give preliminary results on the introduced problems,
focusing on parfactor/subset and parfactor/update PUDs.

\section{Cleaning and Querying Unclean Data}\label{sec:mlipqa}
In this section, we draw connections between data cleaning in the PUD
framework (MLIs) and traditional \e{minimum repairs}. We also give
preliminary results on the complexity of cleaning. Finally, we draw a
connection between probabilistic query answering and certain answers.

\subsection{Generalizing Minimum Repairs}

We now show that the concept of an MLI in parfactor/subset PUDs
generalizes the concept of a minimum subset repair, and the concept of
an MLI in parfactor/update PUDs generalizes the concept of an optimal
update repair. Minimum subset repairs correspond to MLIs of PUDs with
hard (or heavy) constraints. Minimum update repairs correspond to MLIs
over PUDs that assume both hard (or heavy) constraints and assumptions
of independence among the attributes. From the viewpoint of
computational complexity, this means that finding an exact MLI is not
easier than finding a minimum repair, which is often computationally
hard~\cite{DBLP:conf/pods/LivshitsKR18}. Therefore, we should aim for
\e{approximation} guarantees (which have clear semantics in the
probabilistic setting) if we wish to avoid restricting the generality
of the input.

\subparagraph*{Subset repairs and parfactor/subset PUDs} Recall that
in parfactor/subset PUDs (as defined in
Section~\ref{sec:pudexamples}), every intention $I$ with a nonzero
probability is a subset of the observed unclean database $J\str$. In
particular, every MLI is subset of $J\str$.  Our first result relates
cleaning in parfactor/subset PUDs to the traditional minimum subset
(or \e{cardinality}) repairs. This result states, intuitively, that
the notion of an MLI in a parfactor/subset PUD coincides with the
notion of a minimum subset repair if the weight of the formulas is
high enough and the probability of introducing error is small enough.

\def\thmmlisubset{ Let $(\D,\tau,J\str)$ be a parfactor/subset PUD
  with $\D=(\K,\Phi,w)$. For $i\in\tids(J\str)$, assume that
  $\K[i](\bot)>0$ and $\tau[i](J\str[i])>0$, let
  $q(i)=\K[i](J\str[i])/(\K[i](\bot)\cdot \tau[i](J\str[i]))$, and
  assume that $q(i)\geq 1$.  There is a number $M$ such that if
  $w(\varphi)>M$ for all $\varphi\in\Phi$ then the following are
  equivalent for all $I\subseteq J^\star$:
\begin{enumerate}
\item $I$ is an MLI.
\item $I$ is a minimum subset repair of $J\str$ w.r.t.~the weight
  function $w(i)=\log(q(i))$.
\end{enumerate}
}

\begin{theorem}\label{thm:mli-subset}
\thmmlisubset
\end{theorem} 

\noindent Note that in the theorem, $q(i)$ is the ratio between
$\K[i](J\str[i])$, namely the probability that $\K$ produces the $i$th
tuple of $J\str$, and $\K[i](\bot)\cdot \tau[i](J\str[i])$, namely the
probability that $\K$ does not generate the $i$th tuple of $J\str$ but
$\tau$ does.

Next, we draw a similar connection between minimum update repairs and
MLIs of parfactor/update PUDs.

\subparagraph*{Update repairs and parfactor/update PUDs} We now turn
our attention to update repairs. Recall that in a parfactor/update PUD
(defined in Section~\ref{sec:pudexamples}), the intended clean
database $I$ is assumed to be an update of the observed unclean
database $J\str$. We establish a result analogous to
Theorem~\ref{thm:mli-subset}, stating conditions under which MLIs for
parfactor/update PUDs coincide with traditional minimum update
repairs.

Let $\scs$ be a schema, and let $\U=(\D,\kappa,J\str)$ be a
parfactor/update PUD with $\D=(\K,\Phi,w)$. We say that $\U$ is
\e{attribute independent} if $\K$ and $\kappa$ feature probabilistic
independence among the attributes. More precisely, if $i\in R^{\D}$
for $R\in\scs$ with $\sig(R)=(A_1,\dots,A_k)$, then we assume that
$\K[i](a_1,\dots,a_k)$ can be written as
$\K[i](a_1,\dots,a_k)=\prod_{j=1}^k \K_{A_j}[i](a_j)$ and, for
$t=(b_1,\dots,b_k)$, that $\kappa[i,t](a_1,\dots,a_k)$ can be written
as $\kappa[i,t](a_1,\dots,a_k)=\prod_{j=1}^k
\kappa_{A_j}[i,b_j](a_j)$.  In particular, the choice of the value
$a_j$ depends only on $b_j$ and not on other values $b_{j'}$.  

The following theorem states that the concept of an MLI of a
parfactor/update PUD coincides with the concept of a minimum update
repair when the PUD is attribute independent and, moreover, the weight
of the integrity constraints is high.

\def\thmmliupdateold{ ~Let $(\D,\vg,p,J\str)$ be an MLD/update PUD
  where $\D$ is the MLD $(\P,\Phi,w)$ and $\P$ is the TIID
  $(\varrho,p',\tg')$.  Let $\tg'$ be symmetric and able to freely
  update $J\str$.  There is a number $M$ such that if
  $w_{\varphi}>M$ for all $\varphi\in\Phi$ then the following are
  equivalent for all updates $I$ of $J^\star$: (1) $I$ is an MLI, and
  (2) $I$ is a minimum update repair of $J\str$ w.r.t.~$\Phi$ under
  the value weight $w^{R.A}_{J\str}(i)=-\log(r_{R.A}(i))$.  }

\def\thmmliupdate{ Let $\U=(\D,\tau,J\str)$ be an
  attribute-independent parfactor/update PUD with $\D=(\K,\Phi,w)$
  such that $\U\str(I)>0$ for at least one consistent update $I$ of
  $J\str$.  There is a number $M$ such that if $w(\varphi)>M$ for all
  $\varphi\in\Phi$, then the following statements are equivalent for
  all $I\subseteq J^\star$:
\begin{enumerate}
\item $I$ is an MLI.
\item $I$ is a minimum update repair w.r.t.~the weight function
  \[w(i.A,a)=-\log(\K_A[i](a)\cdot\kappa_A[i,a](J\str[i].A))\,.\]
\end{enumerate}
}
\begin{theorem}\label{thm:mli-update}
\thmmliupdate
\end{theorem} 
Note that $\K_A[i](a)\cdot\kappa_A[i,a](J\str[i].A)$ is the
probability that $a$ is produced by $\K$ for the cell $i.A$, and that
$a$ is then changed to $J\str[i].A$ via $\kappa$. Also note that in
the case where this product is zero, we slightly abuse the notation by
assuming that the weight is infinity.

\subsection{Complexity of Cleaning with Key Constraints}
We now present a complexity result on computing an MLI of a
parfactor/subset PUD in the presence of key constraints. The following
theorem states that in the case of a single key constraint per
relation (which is the common setup, e.g., for the analysis of certain
query
answering~\cite{DBLP:journals/tods/KoutrisW17,DBLP:conf/icde/AndritsosFM06,DBLP:journals/ipl/KolaitisP12}),
an MLI can be found in polynomial time. Note that we do not make any
assumption about the parameters; in particular, it may be the case
that an MLI violates the key constraints since the constraints are
weak. Regarding the representation of the probability spaces $\K$ and
$\tau$, the only assumption we make is that, given a tuple $t$, the
probabilities $\K[i](t)$ and $\tau[i](t)$ can be computed in
polynomial time.

\def\thmmlikey{ Let $(\D,\tau,J\str)$ be a parfactor/subset PUD with
  $\D=(\K,\Phi,w)$. If $\Phi$ consists of (at most) one key constraint
  per relation, and no relation of $J\str$ has duplicate tuples, then
  an MLI can be computed in polynomial time.  }

\begin{theorem}\label{thm:mli-key}
\thmmlikey
\end{theorem} 

It is left for future investigation to seek additional constraints
(e.g., functional dependencies) for which an MLI can be found in
polynomial time.  Note that Theorem~\ref{thm:mli-subset} implies that
(under conventional complexity assumptions) we cannot generalize the
polynomial-time result to all sets of functional dependencies, since
finding a minimum subset repair might be computationally
hard~\cite{DBLP:conf/icdt/LopatenkoB07,DBLP:conf/pods/LivshitsKR18}.

\subsection{Probabilistic Query Answering}
\label{sec:pqa}
For probabilisitic query answering, we again focus on the
parametric/subset PUDs, and now we draw a connection to \e{consistent
  query answering} over the cardinality repairs.  Recall that a
\e{consistent answer} for a query $Q$ over an inconsistent database
$J$ is a tuple $t$ that belongs to $Q(I)$ for every cardinality repair
$I$ of $J$.

Let $\scs$ be a schema, $J\str$ a database, and $\Phi$ a set of
integrity constraints. Let $M=|\tids(J\str)|$.  The \e{uniform}
parfactor/subset PUD for $J\str$ and $\Phi$ with the parameters $p$
and $u$, denoted $\U_{p,u}(J\str,\Phi)$ or just $\U_{p,u}$ if $J\str$
and $\Phi$ are clear from the context, is the parfactor/subset PUD
$(\D,\tau,J\str)$ with $\D=(\K,\Phi,w)$ such that the following hold.
\begin{itemize}
\item For all $R\in\scs$ we have $\tids(R^{\K})=\tids(R^{J\str})$ and
  $\K[i](x)=1/M$ for every tuple identifier $i$ and argument $x$ in
  $R^{\K}\cup\set{\bot}$. The remaining mass (required to reaching
  $1$) is given to an arbitrary tuple outside of $J\str$.
\item $w(\varphi)=u$ for every $\varphi\in\Phi$.
\item $\tau[i](\bot)=p$, and $\tau[i](t)=(1-p)/M$ for every identifier
  $i$ and tuple $t$ in $R^{\K}$. Again, the remaining mass is given to
  an arbitrary tuple outside of $J\str$.
\end{itemize}
Observe that $\D$ is defined in such a way that every subset $I$ of
$J\str$ has the same prior probability $\K(I)$, namely
$1/M^{|J\str|}$.

The following theorem states that, for $\U_{p,u}$, the consistent
answers are precisely the answers whose probability approaches one
when all of the following hold:
\e{(1)}
the probability of introducing error (i.e., $1-p$) approaches zero;
\e{and (2)} the weight of the weak constraints (i.e., $u$) approaches
infinity, that is, the constraints strengthen towards hardness.

\def\thmpqasubset{ Let $J\str$ be a database, $\Phi$ a set of
  integrity constraints, $Q$ a query, and $t$ a tuple.  The following
  are equivalent: 
\begin{enumerate}
\item $t$ is a consistent answer over the cardinality repairs.
\item $\lim_{p\limit 1}\lim_{w\limit\infty}\prob_{I\sim\U_{p,u}\str}(t\in
  Q(I))=1$.  
\end{enumerate}
}

\begin{theorem}\label{thm:pqa-subset}
\thmpqasubset
\end{theorem} 

\noindent Therefore, Theorem~\ref{thm:pqa-subset} sheds light on the
role that the consistent answers have in probabilisitic query
answering over parfactor/subset PUDs.

\section{Learning Probabilistic Unclean Databases}
\label{sec:learning}

We now give preliminary results on PUD learning, focusing on
parfactor/update PUDs.  We begin by describing the setup we consider
in this section and the representation system $\pudl$ we use to
describe the parametric intention and realization models $\I_{\Xi}$
and $\R_{\Theta}$.

\subsection{Setup}
\label{sec:learningsetup}
To discuss the learning of parameters, we need to specify the actual
parametric model we assume.  Let $\scs$ be a schema, and let
$\U=(\D,\kappa,J\str)$ be a parfactor/update PUD with
$\D=(\K,\Phi,w)$.  Since we restrict the discussion to
parfactor/update PUDs, we assume (without loss of generality) that the
identifiers in $\D$ are exactly those in $J\str$, that is,
$\tids(\D)=\tids(J\str)$.

In our setup, $\K$ of $\D$ and $\kappa$ of $\U$ are expressed in a
parametric form that allows us to define the parametric intention and
realization models $\I_{\Xi}$ and $\R_{\Theta}$ as \e{Gibbs}
distributions.  We refer to these as {\em Gibbs parfactor/update PUD
  models}, and define them as follows.

\subparagraph*{Parametric intention} To specify $\K$, we assume that
for each relation symbol $R \in \scs$, the probability $\K[i](t)$,
with $i \in \tids(R^{\D})$, is expressed in the form of an exponential
distribution $\K[i](t) = \exp\left(\sum_{f \in F}w_f f(t)\right)$
where each $f\in F$ is an arbitrary function (\e{feature}) over $t$,
and each weight $w_f$ is a real number.  

For example, a feature $f\in F$ may be a function that takes as input
a tuple and returns a value in $\{-1,1\}$. An example feature $f$ can
state that $f(t)=1$ if $t[\textsf{gender}]=\texttt{female}$ and,
otherwise, $f(t)=-1$.  Another example is $f[t]=1$ if
$t[\textsf{zip}]$ starts with \texttt{53} and
$t[\textsf{state}]=\texttt{WI}$, and otherwise $f[t]=-1$.  Additional
examples of such features include the ones used in our prior work on
HoloClean~\cite{holoclean} to capture the co-occurrence probability of
attribute value pairs. Each weight $w_f$ corresponds to a parameter of
the model. An assignment to these weights gives as a probability
distribution, similarly to probabilistic graphical
models~\cite{Koller:2009:PGM:1795555}.

The parameter vector $\Xi$ of the parametric intention model
$\I_{\Xi}$ consists of two sets of parameters: the weights $w_f$ for
each feature $f \in F$, and the weights $w(\varphi)$, which we write
as $w_\varphi$ for uniformity of presentation, for each constraint
$\varphi \in \Phi$.  Thus, the overall parametric intention model
$\I_{\Xi}$ is expressed as a parametric Gibbs distribution.

We will take a special interest in the case where the integrity
constraints are \e{unary}, which means that they have the form
$\forall{x}[\gamma(x)]$, where $\gamma$ is quantifier free; hence, a
unary constraint is a statement about a single tuple.  Examples of
unary constraints are restricted cases of conditional functional
dependencies~\cite{DBLP:conf/icde/BohannonFGJK07,
  Fan:2008:CFD:1366102.1366103}. An example of such a constraint can
be ``age smaller than 10 cannot co-occur with a salary greater than
\$100k.''

\subparagraph*{Parametric realization} We consider a parametric realization
model that is similar to the parametric intention model presented
above, which is again a parametric Gibbs distribution. For each
relation symbol $R \in \scs$ and every pair $(t, t^\prime)\in
\tuples(R)\times\tuples(R)$, the probability $\kappa[i,t](t^\prime)$,
with $i \in \tids(\D)$, is expressed in the form of the Gibbs
distribution $ \kappa[i,t](t^\prime) =
\frac{1}{Z_{\kappa}(t)}\exp\left(\sum_{g \in G}w_g
  g(t,t^\prime)\right)$, where $G$ is the set of features, and each $g$ is
an arbitrary function (\e{feature}) over $(t,t^\prime)$, each weight
$w_g$ is a real number,
and $Z_{\kappa}(t)$ is a normalization constant defined as $
Z_{\kappa}(t) = \sum_{t^\prime \in \tuples(R)}\exp\left(\sum_{g \in
    G}w_g g(t,t^\prime)\right) $.  Hence, we get a parametric model
for the probability distribution $\kappa$.

As an example, a feature function $g(t,t^\prime)$ may capture spelling
errors: $g(t,t^\prime)=1$ if $t^\prime[\mathsf{city}]$ can be obtained
by deleting one character from $t[\mathsf{city}]$, and otherwise,
$g(t,t^\prime)=-1$. The parameter vector $\Theta$ of the parametric
realization model $\R_{\Theta}$ consists of the set of weights $w_g$
for each feature $g \in G$.

\subparagraph*{Assumptions} We make two assumption here. First, we assume
that the attributes of all relation symbols in $\scs$ take values over
a finite and given set. This means that for each relation symbol $R
\in \scs$, $\tuples(R)$ is also finite and given as input. Second, all
features describing $\K$ and $\kappa$ can be computed efficiently,
that is, in polynomial time in the size of the input.

\subparagraph*{Obtaining examples for learning} For supervised
learning, we require a training collection $(I_j,J_j)_{j=1}^n$ and for
unsupervised learning, a training collection $(J_j)_{j=1}^n$.  We
would like to make the case that, oftentimes, a single large example
can be broken down into many small examples.  Recall that in our setup
(Gibbs parfactor/update PUDs), cross-tuple correlations can be
introduced only by the integrity constraints in $\Phi$.  Consider, for
instance, the case where all constraints in $\Phi$ are unary.  Then,
we get \e{tuple-independent} parfactor/update PUDs, and each tuple
identifier $i$ can become an example database: $(I[i],J[i])$ in the
supervised case, and $J[i]$ in the unsupervised case.  For general
constraints, cross-tuple correlations exist, and each example
$(I_j,J_j)$ and $(J_j)$ can be obtained by taking correlated groups of
tuples from $J\str$ by considering different values for the attributes
participating in each constraint $\varphi \in \Phi$. The number of
tuples contained in each example depends on the constraints. This is a
standard practice with parameterized probabilistic models as the ones
we consider here~\cite{pmlr-v28-london13}.

\subsection{Supervised Learning}
\label{sec:superlearn}

We begin by considering supervised $(\scs,\pudl)$-learning. All
results presented in this section build upon standard tools from
statistical learning.  We are given a collection $(I_j,J_j)_{j=1}^n$
of intention-realization examples, and the goal is to find parameter
values $\mathbf{c}$ and $\mathbf{d}$ that maximize the likelihood of
pairs $(I_j,J_j)_{j=1}^n$, that is,
$\prod_{j=1}^n\tightcirc{\R}{\I}(I_j,J_j)$ for
$\I=\I_{\Xi/\mathbf{c}}$ and $\R=\R_{\Theta/\mathbf{d}}$.  To
facilitate the analysis, we write this objective function as a sum
over terms by considering the negative log-likelihood $l(\Xi =
\mathbf{c},\Theta = \mathbf{d}; (I_j,J_j)_{j=1}^n) =
-\log\prod_{j=1}^n\tightcirc{\R}{\I}(I_j,J_j) =
-\sum_{j=1}^n\log\tightcirc{\R}{\I}(I_j,J_j)$, and seek parameter
values $\mathbf{c}$ and $\mathbf{d}$ that minimize it.  For
parfactor/update PUDs we have that $\tightcirc{\R}{\I}(I,J)=
\D(I)\times \prod_{\substack{i\in\tids(I)}} \kappa[i,I[i]](J[i])$
where $\D(I) = \K(I)\,\,\times\,\,
1/Z\,\,\times\,\,
\exp( -\sum_{\varphi\in\Phi} w(\varphi)\times |V(\varphi,I)|)$. For
the Gibbs parfactor/update PUD, $\mathbf{c}$ corresponds to an
assignment of parameters $w_{f}$ describing $\K$ and parameters
$w_\varphi = w(\varphi)$ for the constraints
$\varphi\in\Phi$. Similarly, $\mathbf{d}$ corresponds to an assignment
of parameters $w_{g}$ describing $\kappa$. We use $Z(\mathbf{c})$ to
denote the partition function $Z$ under the parameters
$\mathbf{c}$. We have:
\begin{align}
\notag
l(\mathbf{c},\mathbf{d}; (I_j,J_j)_{j=1}^n) = &-\sum_{j=1}^n\log \left(\K(I_j;\mathbf{c})\,\,\times\,\,
\exp\left(
-\sum_{\varphi\in\Phi}
w_\varphi\times |V(\varphi,I_j)|\right)\right) + n\log Z(\mathbf{c})\\
& - \sum_{j=1}^n\sum_{\substack{i\in\tids(I_j)}} \log\left(\kappa[i,I[i];\mathbf{d}](J_j[i])\right)
\label{eq:l}
\end{align}
where $\K(I_j;\mathbf{c})$ denotes that $\K$ is parametrized by
$\mathbf{c}$ and $\kappa[i,I[i];\mathbf{d}]$ denotes that $\kappa$ is
parametrized by $\mathbf{d}$.

Our goal becomes to minimize the expression in~\eqref{eq:l}. It is
well-known from the ML literature that there is no
analytical solution to such minimization problems, and one needs to
use iterative \e{gradient-based}
methods~\cite{Koller:2009:PGM:1795555}. We investigate whether
gradient-based methods can indeed find a global minimum, and whether
computing the gradient of this objective during each iteration is
tractable. For the first question, the answer is positive.

\def\propconvex{
  $l(\Xi = \mathbf{c},\Theta = \mathbf{d}; (I_j,J_j)_{j=1}^n)$ is a
  convex function of  $(\Xi,\Theta)$.  }

\begin{proposition}\label{prop:convex_likelihood}
  \propconvex
\end{proposition}
This proposition implies that the optimization objective for
supervised learning has only global optima. Hence, it is guaranteed
that any gradient-based optimization method will converge to a global
optimum. Next, we study when the gradient of $l(\Xi =
\mathbf{c},\Theta = \mathbf{d}; (I_j,J_j)_{j=1}^n)$ with respect to
$\mathbf{c}$ and $\mathbf{d}$ can be computed efficiently.

To compute the gradient of $l(\Xi = \mathbf{c},\Theta = \mathbf{d};
(I_j,J_j)_{j=1}^n)$, one has to compute $\frac{\partial}{\partial
  c_{l}}\log Z(\mathbf{c})$. To compute this derivative a full
inference step is required~\cite{Koller:2009:PGM:1795555}. This is
because computing this gradient amounts to computing the expected
value for each feature (corresponding to each parameter $c_{l}$)
according to the distribution defined by
$\mathbf{c}$~\cite[Proposition~20.2]{Koller:2009:PGM:1795555}. However,
marginal inference is often
\#P-hard~\cite{Globerson:2015:HIS:3045118.3045350}. In our setup, the
constraints in $\Phi$ correspond to features, and it is not clear
whether the gradient can be efficiently computable.  In general, one
can still estimate the aforementioned gradient by using approximate
inference methods such as \e{Markov Chain Monte Carlo} (MCMC)
methods~\cite{doi:10.1080/01621459.1993.10594284,Singh:2012:MCM:2390948.2391072}
or \e{belief propagation}~\cite{Wainwright03}. While effective in
practice, these methods do not come with guarantees on the quality of
the obtained solution.  Next, we focus on an instance of PUD learning
where exact inference is tractable (linear on $\tuples(R)$), hence, we
can compute the exact gradients of the aforementioned optimization
objective efficiently.

\subparagraph*{Tuple independence} We focus on Gibbs parfactor/update
PUD models where all constraints in $\Phi$ are unary. Here,
$(I_j,J_j)_{j=1}^n$ corresponds to a collection of $n$ examples, each
of which having one tuple identifier. We use $I_j[0]$ and $J_j[0]$ to
denote the tuples in the example $(I_j, J_j)$. In the Appendix, we show that
the negative log-likelihood $l(\Xi = \mathbf{c},\Theta = \mathbf{d};
(I_j,J_j)_{j=1}^n)$ factorizes as $l(\mathbf{c},\mathbf{d};
(I_j,J_j)_{j=1}^n) = \sum_{j=1}^n l^\prime(\mathbf{c},\mathbf{d};
I_j[0],J_j[0]),$ where each $l^\prime(\mathbf{c},\mathbf{d};
I_j[0],J_j[0])$ is a convex function of $\Xi$ and $\Theta$ (see Proposition~\ref{prop:decomp} in the Appendix). Moreover, we show
that the gradient of each $l^\prime(\mathbf{c},\mathbf{d};
I_j[0],J_j[0])$ can be evaluated in time linear to $\tuples(R)$ where
$R$ is the relation corresponding to the tuple identifier associated
with example $(I_{j},J_{j})$ (see Proposition~\ref{prop:linear} in the Appendix). Hence, we get the
following.

\def\thmconvexindep{Given a training collection $(I_j,J_j)_{j=1}^n$
  and a Gibbs parfactor/update PUD model with unary constraints, the
  exact gradient of $l(\Xi = \mathbf{c},\Theta = \mathbf{d};
  (I_j,J_j)_{j=1}^n)$ can be evaluated in $O(n\cdot \max_{R \in
    \scs}|\tuples(R)|)$ time.}

\begin{theorem}\label{thm:convex_indep}
\thmconvexindep
\end{theorem}
The above theorem implies that convex-optimization techniques such as
\e{stochastic gradient descent}~\cite{Boyd:2004:CO:993483} can be used
to scale to large PUD learning instances (i.e., for large $n$).

A question that arises is about the number of examples (i.e., $n$)
required to learn a PUD model. To answer this question, we study the
convergence of \e{supervised $(\scs,\pudl)$-learning} for Gibbs
parfactor/update PUD models with unary constraints. For that, we view
the collection $(I_j,J_j)_{j=1}^n$ as independent and identically
distributed (i.i.d.)~examples, drawn from a distribution that
corresponds to a Gibbs parfactor/update PUD model with unary
constraints and true parameters $\mathbf{c}^\star$ and
$\mathbf{d}^\star$. By the law of large numbers, the \e{Maximum
  Likelihood Estimates} (MLE) $\mathbf{c}$ and $\mathbf{d}$ are
guaranteed to converge to $\mathbf{c}^\star$ and $\mathbf{d}^\star$ in
probability. This means that for arbitrarily small $\epsilon >0 $ we
have that $P(|\mathbf{c} - \mathbf{c}^\star| > \epsilon) \rightarrow
0$ as $n \rightarrow \infty$. The same holds for
$\mathbf{d}$. Moreover, we show that the MLE $\mathbf{c}$ and
$\mathbf{d}$ satisfy the property of \e{asymptotic
  normality}~\cite{kullback1997information}.  Intuitively, asymptotic
normality states that the estimator not only converges to the unknown
parameter, but it converges fast enough at a rate of
$1/\sqrt{n}$. This implies that to achieve the error $\epsilon$ for
$\mathbf{c}$ and $\mathbf{d}$, one only needs $n = O(\epsilon^{-2})$
training examples.

\def\thmconverge{Consider a training collection $(I_j,J_j)_{j=1}^n$
  drawn i.i.d.~from a Gibbs parfactor/update PUD model with unary
  constraints and true parameters $\mathbf{c}^\star$ and
  $\mathbf{d}^\star$. The maximum likelihood estimates $\mathbf{c}$ and $\mathbf{d}$ satisfy \e{asymptotic normality}, that is,
\[
	\sqrt{n} \left( \mathbf{c} - \mathbf{c}^\star \right) \rightarrow \mathcal{N}\left(0,  \Sigma^2_{\mathbf{c}^\star} \right) \text{as~} n \rightarrow \infty \text{~and~}
	\sqrt{n} \left( \mathbf{d} - \mathbf{d}^\star \right) \rightarrow \mathcal{N}\left(0, \Sigma^2_{\mathbf{d}^\star} \right) \text{as~} n \rightarrow \infty
\]
where $\Sigma^2_{\mathbf{c}^\star}$ and $\Sigma^2_{\mathbf{d}^\star}$
are the asymptotic variance of the estimates $\mathbf{c}$ and
$\mathbf{d}$.}

\begin{theorem}\label{thm:converge}
\thmconverge
\end{theorem}
Note that both the \e{multivariate Gaussian} distribution
$\mathcal{N}(\mathbf{\mu},\Sigma^2)$ and the \e{asymptotic variance}
are defined in classic statistics
literature~\cite{kullback1997information}.

\subsection{Unsupervised Learning}
\label{sec:learndirty}
We now present preliminary results for unsupervised
$(\scs,\pudl)$-learning. We are given a training collection
$(J_j)_{j=1}^n$ and seek to find $\mathbf{c}$ and $\mathbf{d}$ that
minimize the negative log-likelihood $l(\Xi =
\mathbf{c},\Theta=\mathbf{d}; (J_j)_{j=1}^n) =
-\sum_{j=1}^n\log\sum_{I}\tightcirc{\R}{\I}(I,J_j)$.  Again, there is
no analytical solution for finding optimal $\mathbf{c}$ and
$\mathbf{d}$. Hence, one needs to use iterative gradient-based
approaches, and again the questions of convexity and gradient
computation arise.  In the general case, this function is \e{not}
necessarily convex. Hence, gradient-based methods are not guaranteed
to converge to a global optimum. However, one can still solve the
corresponding optimization problem using non-convex optimization
methods~\cite{bertsekas1999nonlinear}. Nevertheless, we show next that
when realizers do not introduce \e{too much error}, we can establish
guarantees.

\subparagraph*{Low-noise condition} Consider a Gibbs parfactor/update PUD
model. We say that a PUD defined by $\I_{\Xi/\mathbf{c}}$ and
$\R_{\Theta/\mathbf{d}}$ satisfies the \e{low-noise condition} with
probability $p$ if the realizer introduces an
error with probability at most $p$. That is, for all intensions $I$
and identifiers $i\in\tids(I)$, it is the case that $\prob(J[i] =
I[i]\mid I) \ge 1 - p$. We have the following.

\def\convexnll{ Consider a Gibbs parfactor/update PUD model where
  $\Xi$ takes values from a compact convex set. Given a training
  collection $(J_j)_{j=1}^n$, there exists a fixed probability $p > 0$
  such that, under the low-noise condition with probability $p$,
  the negative log-likelihood $l(\Xi = \mathbf{c},\Theta=\mathbf{d};
  (J_j)_{j=1}^n)$ is a convex function of $\Xi$.  }

\begin{theorem}\label{thm:convexnll}
\convexnll
\end{theorem}
Hence, in certain cases, it is possible to find a global optimum of
the overall negative log-likelihood. For that, the low-noise condition
should hold with probability $p$ that is also bounded, that is, it
cannot be arbitrarily large.  We can show that, if the low-noise
condition holds with probability $p$, it is indeed bounded for Gibbs
parfactor/update PUDs with unary constraints (see Proposition~\ref{prop:bounded} in the Appendix). We then find a
global optimum as follows.  We assume a simple parametric realization
$\R_{\Theta}$, that is, a model for which we can efficiently perform
\e{grid search} over the space of parameter values $\mathbf{d}$. To
find the global optimum for the negative log-likelihood, we solve a
series of convex optimization problems over $\Xi$ for different fixed
$\Theta = \mathbf{d}$. For each of these problems, we are guaranteed
to find a corresponding global optimum $\mathbf{c}$, and by performing
a grid search we are guaranteed to find the overall global optima
$\mathbf{c}$ and $\mathbf{d}$. This approach has been shown to
converge for similar simple non-convex
problems~\cite{conf/icml/SaRO15,ma2017implicit}.

Finally, similarly to supervised learning, the negative log-likelihood
for fixed $\Theta=\mathbf{d}$ decomposes into a sum of convex losses
over $(J_j)_{i=1}^n$ where each example $J_j$ contains a single
tuple. We use $J_j[0]$ to denote that tuple. We have that $l(\Xi =
\mathbf{c}, \Theta=\mathbf{d}, (J_j)_{j=1}^n) = \sum_{j=1}^n
l^{\prime}(\mathbf{c},\mathbf{d};J_j[0])$ where $\mathbf{d}$ is
fixed. We show that the gradient of each
$l^{\prime}(\mathbf{c},\mathbf{d};J[0])$ can be evaluated in polynomial time to $\tuples(R)$ where $R$ is the relation corresponding to
the tuple identifier associated with the example $(J_j)$.

\def\thmconvexunsup{Given a training collection $(J_j)_{j=1}^n$ and a
  Gibbs parfactor/update PUD model with unary constraints that
  satisfies the low-noise condition with probability $p$, the exact
  gradient of $l(\Xi = \mathbf{c},\Theta = \mathbf{d}; (J_j)_{j=1}^n)$
  with respect to $\Xi$ can be evaluated in time $O(n\cdot \max_{R \in
    \scs}|\tuples(R)|^2)$ for every fixed $\mathbf{d}$.}

It is left for future work to find sufficient conditions for $p$ to be
bounded for PUD models with more general constraints, as well as the
complexity and convergence aspects.

\section{Concluding Remarks}
\label{sec:conclusion}
Taking inspiration from our experience with the HoloClean
system~\cite{holoclean}, we introduced the concept of Probabilistic
Unclean Databases (PUDs), a framework for unclean data that follows a
noisy channel approach to model how errors are introduced in data. We
defined three fundamental problems in the framework: cleaning,
probabilistic query answering, and PUD learning (parameter
estimation).  We introduced PUD instantiations that generalize the
deterministic concepts of subset repairs and update repairs, presented
preliminary complexity, convergence, and learnability results.

This paper opens up many research directions for future
exploration. One is to investigate the complexity of cleaning in more
general configurations than the ones covered here. Moreover, in cases
where probabilistic cleaning is computationally hard, it is of natural
interest to find \e{approximate} repairs that have a probability
(provably) close to the maximum. Another direction is the complexity
of probabilistic query answering and approximation thereof, starting
with the most basic constraints (e.g., primary keys) and queries
(e.g., determine the marginal probability of a fact). Finally, an
important direction is to devise learning algorithms for cases beyond
the ones we discussed here. In particular, it is of high importance to
understand when we can learn parameters without training data, based
only on the given dirty database, under more general noisy realization
models than the ones discussed in this paper.

\newpage
\bibliographystyle{abbrv}
\bibliography{prob_unclean}

\appendix\label{appendix}
\def\Const{\mathrm{Const}}

\section{Proofs}
\noindent In this section, we provide proofs that are missing from the body
of the paper.

\subsection{Proof of Theorem~\ref{thm:mli-subset}}
\begin{reptheorem}{\ref{thm:mli-subset}}
	\thmmlisubset
      \end{reptheorem}
\begin{proof}
  We analyze the probability of an MLI $I$.  We have the following for
  all subsets $I$ of $J\str$.
\begin{equation}\label{eq:riijstr-subset}
  \tightcirc{\R}{\I}(I,J\str) \,\,=\,\, \D(I)
\times\hskip-0.5em
\prod_{\substack{i\in\tids(J\str)\setminus\\\tids(I)}}\hskip-1em\tau[i](J\str[i])
\times\hskip-0.5em
\prod_{\substack{i\in\tids(\D)\setminus\\\tids(J\str)}}\hskip-1em\tau[i](\bot) 
\,\,\sim\,\,
\D(I)
\times\hskip-0.5em
\prod_{\substack{i\in\tids(J\str)\setminus\\\tids(I)}}\hskip-1em\tau[i](J\str[i])
\end{equation}
The proportionality is due to the fact that the third factor in the
first multiplication is the same for all $I$. If $w_\varphi$ is large
enough for all $\varphi\in\Phi$, then we can assume that every $I$
that satisfies $\Phi$ has a higher probability than every $I$ that
violates $\Phi$. Observe that our integrity constraints and assumption
that $\K[i](\bot)>0$ and $\tau[i](J\str[i])>0$ implies that at least
one consistent intention has a nonzero probability---the \e{empty}
database.  Hence, we can assume to begin with that the MLI $I$ is
selected from the subsets of $J\str$ that satisfy $\Phi$, and
therefore $\D(I)\sim\K(I)$. Moreover, we have the following.
\begin{equation}
\D(I)\sim\K(I)=
\prod_{i\in\tids(I)}\K[i](I[i])
\times\hskip-0.5em
\prod_{\substack{i\in\tids(\K)\setminus\\\tids(I)}}\hskip-0.5em\K[i](\bot)
\,\,\sim\,\,
\prod_{i\in I}\K[i](I[i])
\times\hskip-0.5em
\prod_{\substack{i\in\tids(J\str)\setminus\\\tids(I)}}\hskip-0.5em\K[i](\bot)
\label{eq:pi-subset}
\end{equation}
Here, the proportionality is due to the fact that we divide the
probability by the same factor for all $I$. From
Equations~\eqref{eq:riijstr-subset} and \eqref{eq:pi-subset}, we
conclude the following.

\begin{align}
\tightcirc{\R}{\I}(I,J\str)  \,\,\sim\,\,& 
\prod_{i\in I}\K[i](I[i])
\times\hskip-0.5em
\prod_{\substack{i\in\tids(J\str)\setminus\\\tids(I)}}\hskip-1em
\K[i](\bot)\cdot \tau[i](J\str[i])
\notag
\\=\,\,&
\prod_{i\in I}
\frac
{\K[i](I[i])}
{\K[i](\bot)\cdot \tau[i](J\str[i])}
\times
\prod_{\substack{i\in\tids(J\str)}}\hskip-1em
\K[i](\bot)\cdot \tau[i](J\str[i])
\label{eq:subset-i-j-star}
\\\sim\,\,&
\prod_{i\in I}
\frac
{\K[i](I[i])}
{\K[i](\bot)\cdot \tau[i](J\str[i])}
\,\,=\,\,
\prod_{i\in I}
\frac
{\K[i](J\str[i])}
{\K[i](\bot)\cdot \tau[i](J\str[i])}
\,\,=\,\,
\prod_{i\in I}q(i)
%
\label{eq:max-subset}
\end{align}
Again, the last proportionality is due to the fact that the right
factor in~\eqref{eq:subset-i-j-star} is the same for all $I$.  The
second equality in~\eqref{eq:max-subset} is due to the fact that $I$
is a subset of $J\str$, and so, $I[i]=J\str[i]$ for all
$i\in\tids(I)$.  Hence, $I$ is a most likely repair if and only if $I$
satisfies $\Phi$ and maximizes~\eqref{eq:max-subset}, which is the
same as maximizing the sum $\sum_{i\in\tids(I)}\log q(i)$, as claimed.
\end{proof}

\subsection{Proof of Theorem~\ref{thm:mli-update}}\label{sec:proofs_cleaning}
\begin{reptheorem}{\ref{thm:mli-update}}
\thmmliupdate
\end{reptheorem}
\begin{proof}
Since we consider only updates $I$ of $J\str$, we can ignore all of
the tuples $i\in\tids(\D)\setminus J\str$. Hence, we have the
following.
\begin{equation}\label{eq:riijstr-update}
  \tightcirc{\R}{\I}(I,J\str) \,\,=\,\, \D(I)
\times\hskip-0.5em
\prod_{i.A\in\cells(J\str)}\hskip-1em\kappa_{A}[i,I[i].A](J\str[i].A)
\end{equation}
If $w_\varphi$ is large enough for all $\varphi\in\Phi$, then we can
assume that every $I$ that satisfies $\Phi$ has a higher probability
than every $I$ that violates $\Phi$. Hence, due to our assumption that
at least one consistent update has a nonzero probability, we can
assume to begin with that the MLI $I$ is selected from the subsets of
$J\str$ that satisfy $\Phi$, and therefore $\D(I)\sim\K(I)$. Hence,
from~\eqref{eq:riijstr-update} we conclude the following.
\begin{align}\label{eq:update2}
  \tightcirc{\R}{\I}(I,J\str) \,\,&\sim\,\,
\prod_{i\in\tids(I)}\K[i](I[i])\times
\prod_{i.A\in\cells(J\str)}\hskip-1em\kappa_{A}[i,I[i].A](J\str[i].A)
\notag
\\
&=\,\,
\prod_{i.A\in\cells(J\str)}\hskip-1em\K_A[i](I[i].A)\cdot\kappa_A[i,I[i].A](J\str[i].A)
\end{align}
The last equation is due to the fact that $I$ is an update of $J\str$
and that $\U$ is an attribute-independent parfactor/update PUD.
Hence, maximizing the probability of $I$ amounts to minimizing the weight function
$w(i.A,a)=-\log(\K_A[i](a)\cdot\kappa_A[i,a](J\str[i].A))$.
\end{proof}
\subsection{Proof of Theorem~\ref{thm:mli-key}}
\begin{reptheorem}{\ref{thm:mli-key}}
	\thmmlikey
\end{reptheorem}

\begin{proof}
  Our goal is to compute a subset $I$ of $J\str$ that maximizes
  $\tightcirc{\R}{\I}(I,J\str)$.  As argued in the proof of
  Theorem~\ref{thm:mli-subset}, we can ignore tuple identifiers not in
  $\tids(J\str)$. By a \e{block} of $J\str$ we refer to the (maximal)
  set of tuples of $J\str$ that share a key value in a single
  relation. By our assumption, every block of more than a single tuple
  is a violation of the corresponding key constraint. With a
  straightforward spelling out of $\tightcirc{\R}{\I}(I,J\str)$ we can
  observe that this probability factorizes across the blocks of
  $J\str$ and their corresponding subsets in $I$, due to probabilistic
  independence between blocks.  In particular, it suffices to find an
  MLI for each block separately, and take the union of the MLIs as our
  solution $I$.  Therefore, we can assume that $J\str$ has a single
  block; that is, all the tuples in $J\str$ belong to the same
  relation, and all of them share the same key.

  We denote our single key constraint by $\varphi$. Recall that
  $w(\varphi)$ is the weight of $\varphi$ in $\D$, and that
  $V(\varphi,I)$ the set of violations of $\varphi$ in $I$.  From our
  assumptions it follows that $|V(\varphi,I)|=|I|\cdot (|I|-1)$.  We
  conclude the following.
\begin{align*}
\tightcirc{\R}{\I}&(I,J\str)
\,\,=\,\,
\D(I)\times
\hskip-1em
\prod_{\substack{i\in\tids(J\str)\setminus\\\tids(I)}}
\hskip-1em
\tau(J\str[i])
\\&=\,\,
\prod_{i\in\tids(I)}\K(I[i])
\times
\exp(-w(\varphi)\cdot|I|\cdot (|I|-1))
\times
\hskip-1em
\prod_{\substack{i\in\tids(J\str)\setminus\\\tids(I)}}
\hskip-1em
\tau(J\str[i])
\\&=\,\,
\prod_{i\in\tids(I)}\frac{\K(I[i])}{\tau(J\str[i])}
\times
\exp(-w(\varphi)\cdot|I|\cdot (|I|-1))
\times
\prod_{\substack{i\in\tids(J\str)}}
\hskip-1em
\tau(J\str[i])
\\&\sim\,\,
\prod_{i\in\tids(I)}\frac{\K(I[i])}{\tau(J\str[i])}
\times
\exp(-w(\varphi)\cdot|I|\cdot (|I|-1))
\\&=\,\,
\prod_{i\in\tids(I)}\frac{\K(J\str[i])}{\tau(J\str[i])}
\times
\exp(-w(\varphi)\cdot|I|\cdot (|I|-1))
\end{align*}
The last equality is due to the fact that $I$ is a subset of $J\str$,
and therefore, $I[i]=J\str[i]$. Denote
$q(i)=\K(J\str[i])/\tau(J\str[i])$. In order to find a subset $I$ of
$J\str$ that maximizes $\tightcirc{\R}{\I}(I,J\str)$, we can consider
every size $\ell=0,\dots,|J\str|$ of $I$ and find a subset $I_\ell$
with $|I_\ell|=\ell$ that maximizes the product
$\prod_{i\in I_\ell}q(i)$. Then, we take the $I_\ell$ with the maximal
$\prod_{i\in I_\ell}q(i)\exp(-w(\varphi)\ell(\ell-1))$. In turn, to
find $I_\ell$ it suffices to select $\ell$ tuple identifiers $i$ with
the maximal $q(i)$.
\end{proof}

\subsection{Proof of
  Theorem~\ref{thm:pqa-subset}}\label{sec:proofs_pqa}
\begin{reptheorem}{\ref{thm:pqa-subset}}
	\thmpqasubset
      \end{reptheorem}
\begin{proof}
  Recall the definition of the unnormalized probability
  $\tightcirc{\R}{\I}(I,J\str)$ of $I$.
\begin{align*}
  \tightcirc{\R}{\I}(I,J\str) \,\,=\,\, \D(I)
\times\hskip-0.5em
\prod_{\substack{i\in\tids(J\str)\setminus\\\tids(I)}}\hskip-1em\tau[i](J\str[i])
\times\hskip-0.5em
\prod_{\substack{i\in\tids(\D)\setminus\\\tids(J\str)}}\hskip-1em\tau[i](\bot) 
\,\,\sim\,\,
\D(I)
\times\hskip-0.5em
\prod_{\substack{i\in\tids(J\str)\setminus\\\tids(I)}}\hskip-1em\tau[i](J\str[i])
\end{align*}
We make the following observations.
\begin{itemize}
\item For any constant $p$, as $u\limit\infty$ the mass of $\D$
  concentrates on the databases consistent subsets of $J\str$ (i.e.,
  the ones that satisfy $\Phi$), or, in other words, as
  $u\limit\infty$ the probability that $I$ satisfies $\Phi$ approaches
  $1$.
\item Since $\tau$ assigns the same probability to every tuple, we can
  write
  $\tightcirc{\R}{\I}(I,J\str)\sim \D(I)\cdot C^{|J\str\setminus I|}$
  for some $C\in(0,1)$.  Moreover, as $p$ approaches $1$, the constant
  $C$ approaches $0$.
\item From our definition of $\D$ and $\K$ it follows that every
  consistent subset has the same probability in $\D$ (i.e.,
  $\D(I)=\D(I')$ whenever $I$ and $I'$ are consistent subsets of
  $J\str$).
\end{itemize}
We conclude that if $I$ is cardinality repair and $I'$ is any subset
that is not a cardinality repair (i.e., inconsistent or not maximal in
cardinality), then the ratio between $\tightcirc{\R}{\I}(I,J\str)$ and
$\tightcirc{\R}{\I}(I',J\str)$, and so the ratio between the
probabilities $\U\str(I)$ and $\U\str(I')$, approaches infinity as $p$
approaches $1$.

We conclude that the total probability of the cardinality repairs
approaches $1$ as $p$ approaches $1$, and all cardinality repairs have
the same probability.  In particular, if $t$ is a consistent answer,
then its probability approaches $1$.  Conversely, if $t$ is \e{not} a
consistent answer, then there is a contant portion of the probability
that is missing for every $p$---this is the probability of a
cardinality repair $I$ in which $t\notin Q(I)$.
\end{proof}

\subsection{Proof of Proposition~\ref{prop:convex_likelihood}}\label{sec:proofs_learning}

\begin{repproposition}{\ref{prop:convex_likelihood}}
\propconvex
\end{repproposition}
\begin{proof}
The negative log-likelihood is:
\begin{align*}
l(\Xi = \mathbf{c},\Theta = \mathbf{d}; (I_j,J_j)_{j=1}^n) = &-\sum_{j=1}^n\log \left(\K(I_j;c)\,\,\times\,\,
\exp\left(
-\sum_{\varphi\in\Phi}
w_\varphi\times |V(\varphi,I_j)|\right)\right) \\
&+ n\log Z(\mathbf{c})\\
&- \sum_{j=1}^n\sum_{\substack{i\in\tids(I_j)}} \log\left(\kappa[i,I_j[i];d](J_j[i])\right)
\end{align*}
Additionally we have that:
\[\K(I)=\prod_{i \in \tids(I)}\K[i](I[i]) = \prod_{i \in \tids(I)}\exp\left(\sum_{f \in F}w_f f(I[i])\right)
\]
and 
\[\kappa[i,I[i]](J[i]) =
\frac{1}{Z_{\kappa}(I[i])}\exp\left(\sum_{g \in G}w_g
  g(I[i],J[i])\right)
\]. 

Replacing these to $l(\mathbf{c},\mathbf{d}; (I_j,J_j)_{j=1}^n)$ we have:
\begin{align*}
l(\Xi = \mathbf{c},\Theta = \mathbf{d}; (I_j,J_j)_{j=1}^n) = &-\sum_{j=1}^n\log \left(\prod_{i \in \tids(I_j)}\exp\left(\sum_{f \in F}w_f f(I_j[i])\right)\right)\\
&+\sum_{j=1}^n
\sum_{\varphi\in\Phi}
w_\varphi|V(\varphi,I_j)|\\
&+ n\log Z(\mathbf{c})\\
&- \sum_{j=1}^n\sum_{\substack{i\in\tids(I_j)}} \log\left(\frac{1}{Z_{\kappa}(I_j[i])}\exp\left(\sum_{g \in G}w_g
  g(I_j[i],J_j[i])\right)\right)
\end{align*}
or
\begin{align*}
l(\Xi = \mathbf{c},\Theta = \mathbf{d}; (I_j,J_j)_{j=1}^n) = &-\sum_{j=1}^n\sum_{i\in \tids{I_j}} \left(\sum_{f \in F}w_f f(I_j[i])\right)\\
&+\sum_{j=1}^n
\sum_{\varphi\in\Phi}
w_\varphi|V(\varphi,I_j)|\\
&+ n\log Z(\mathbf{c})\\
&- \sum_{j=1}^n\sum_{\substack{i\in\tids(I_j)}} \left(\sum_{g \in G}w_g
  g(I_j[i],J_j[i])-\log Z_{\kappa}(I_j[i])\right)
\end{align*}
or
\begin{align*}
l(\Xi = \mathbf{c},\Theta = \mathbf{d}; (I_j,J_j)_{j=1}^n) = &-\sum_{j=1}^n\sum_{i\in \tids{I_j}} \sum_{f \in F}w_f f(I_j[i]) \\
&+\sum_{j=1}^n
\sum_{\varphi\in\Phi}
w_\varphi|V(\varphi,I_j)|\\
&+ n\log \left( \sum_{I} \exp\left(\sum_{i \in \tids(I)}\sum_{f \in F}w_f f(I[i])\right)\right)\\
&+ n\log \left( \sum_{I} \exp\left(-\sum_{\varphi\in\Phi} w_\varphi |V(\varphi,I)|\right)\right)\\
&- \sum_{j=1}^n\sum_{\substack{i\in\tids(I_j)}} \sum_{g \in G}w_g
  g(I_j[i],J_j[i])\\
&+\sum_{j=1}^n\sum_{\substack{i\in\tids(I_j)}}\log \left( \sum_{t^\prime \in \tuples(R^i)}\exp\left(\sum_{g \in
    G}w_g g(I_{j}[i],t^\prime)\right) \right)
\end{align*}
where $R^i$ denotes the relation symbol $R \in \scs$ associated with tuple identifier $i$. Now, recall that vector $\Xi$ contains all weights $w_f$ and $w_\varphi$ while vector $\Theta$ contains all weights $\w_{g}$. We have that the negative log-likelihood is a convex function of $\Xi$ as it corresponds to the sum of LogSumExp functions---those components corresponding to partition functions---with affine functions, which is well-known to be convex~\cite{Boyd:2004:CO:993483}. Similarly for $\Theta$.
\end{proof}

\subsection{Proof of Theorem~\ref{thm:convex_indep}}
Before we prove Theorem~\ref{thm:convex_indep}, we prove two necessary propositions. Recall that we consider Gibbs parfactor/update PUD models with unary constraints. Here, $(I_j,J_j)_{j=1}^n$ is a collection of $n$ i.i.d. examples each of which is associated with one tuple identifier. We use $I_j[0]$ and $J_j[0]$ to denote the tuples in example $(I_j, J_j)$.

First, we show that for Gibbs parfactor/update PUD models with unary constraints the negative log-likelihood $l(\Xi = \mathbf{c},\Theta = \mathbf{d}; (I_j,J_j)_{j=1}^n)$ can be written as $l(\mathbf{c},\mathbf{d}; (I_j,J_j)_{j=1}^n) = \sum_{j=1}^n l^\prime(\mathbf{c},\mathbf{d}; I_j[0],J_j[0]),$ where  each $l^\prime(\mathbf{c},\mathbf{d}; I_j[0],J_j[0])$ is a convex function of $\Xi$ and $\Theta$. Recall that We have:

\def\propdecomp{
Given a collection $(I_j,J_j)_{j=1}^n$ of intention-realization examples and a Gibbs parfactor/update PUD model with unary constraints we have that $l(\mathbf{c},\mathbf{d}; (I_j,J_j)_{j=1}^n) = \sum_{j=1}^n l^\prime(\mathbf{c},\mathbf{d}; I_j[0],J_j[0]),$ where  each $l^\prime(\mathbf{c},\mathbf{d}; I_j[0],J_j[0])$ is a convex function of $\Xi$ and $\Theta$.
}

\begin{proposition}\label{prop:decomp}
\propdecomp
\end{proposition}
\begin{proof}
From the proof of Proposition~\ref{prop:convex_likelihood} we have that negative log-likelihood is:
\begin{align*}
l(\Xi = \mathbf{c},\Theta = \mathbf{d}; (I_j,J_j)_{j=1}^n) = &-\sum_{j=1}^n\sum_{i\in \tids{I_j}} \sum_{f \in F}w_f f(I_j[i]) \\
&+\sum_{j=1}^n
\sum_{\varphi\in\Phi}
w_\varphi|V(\varphi,I_j)|\\
&+ n\log \left( \sum_{I} \exp\left(\sum_{i \in \tids(I)}\sum_{f \in F}w_f f(I[i])\right)\right)\\
&+ n\log \left( \sum_{I} \exp\left(-\sum_{\varphi\in\Phi} w_\varphi |V(\varphi,I)|\right)\right)\\
&- \sum_{j=1}^n\sum_{\substack{i\in\tids(I_j)}} \sum_{g \in G}w_g
  g(I_j[i],J_j[i])\\
&+\sum_{j=1}^n\sum_{\substack{i\in\tids(I_j)}}\log \left( \sum_{t^\prime \in \tuples(R^i)}\exp\left(\sum_{g \in
    G}w_g g(I_{j}[i],t^\prime)\right) \right)
\end{align*}
where $R^i$ denotes the relation symbol $R \in \scs$ associated with tuple identifier $i$.
Given that each example $(I_j, J_j)$ corresponds to one tuple identifier we have that:
\begin{align*}
l(\Xi = \mathbf{c},\Theta = \mathbf{d}; (I_j,J_j)_{j=1}^n) = &-\sum_{j=1}^n \sum_{f \in F}w_f f(I_j[0]) \\
&+\sum_{j=1}^n
\sum_{\varphi\in\Phi}
w_\varphi|V(\varphi,I_j[0])|\\
&+ n\log \left( \sum_{I} \exp\left(\sum_{i \in \tids(I)}\sum_{f \in F}w_f f(I[i])\right)\right)\\
&+ n\log \left( \sum_{I} \exp\left(-\sum_{\varphi\in\Phi} w_\varphi |V(\varphi,I)|\right)\right)\\
&- \sum_{j=1}^n \sum_{g \in G}w_g
  g(I_j[0],J_j[0])\\
&+\sum_{j=1}^n\log \left( \sum_{t^\prime \in \tuples(R^i)}\exp\left(\sum_{g \in
    G}w_g g(I_{j}[0],t^\prime)\right) \right)
\end{align*}
The only components that do not immediately decompose over individual tuple identifiers are $\log \left( \sum_{I} \exp\left(\sum_{i \in \tids(I)}\sum_{f \in F}w_f f(I[i])\right)\right)$ and $\log \left( \sum_{I} \exp\left(-\sum_{\varphi\in\Phi} w_\varphi |V(\varphi,I)|\right)\right)$. We next show that for Gibbs parfactor/update PUD models with unary constraints both can be decomposed into a sum over terms that are defined over individual identifiers.

For the first term we have: By sum separation we have that:
\begin{align*}
\sum_{I} \exp\left(\sum_{i \in \tids(I)}\sum_{f \in F}w_f f(I[i])\right) = \prod_{i \in \tids(\D)} \sum_{t \in \tuples(R^i)} \exp \sum_{f \in F}w_ff(t)
\end{align*}
Hence
\begin{align*}
\log \left( \sum_{I} \exp\left(\sum_{i \in \tids(I)}\sum_{f \in F}w_f f(I[i])\right) \right) &= \log \left (\prod_{i \in \tids(\D)} \sum_{t \in \tuples(R^i)} \exp \sum_{f \in F}w_ff(t) \right) \\
&= \sum_{i \in \tids(\D)}\log \left( \sum_{t \in \tuples(R^i)} \exp \sum_{f \in F}w_ff(t) \right)
\end{align*}

For the second term we have: Since we focus on unary constraints, each formula $\varphi \in \Phi$ has only one free variable. Therefore, we can let $grd(i,\varphi)$ denote that formula grounded with tuple identifier $i \in \tids(\D)$, and rewrite our expression as $\log Z(\mathbf{c}) = \log \left( \sum_{I} \exp\left(-\sum_{\varphi\in\Phi}\sum_{\substack{i \in \tids(I) \\ I\models grd(i,\varphi)}} w_\varphi\right)\right)$. However, it holds that $I\models grd(i,\varphi)$ if and only if $I[i]\models grd(i,\varphi)$. From this we have that $\log \left( \sum_{I} \exp\left(-\sum_{\varphi\in\Phi}\sum_{\substack{i \in \tids(I) \\ I[i]\models grd(i,\varphi)}} w_\varphi\right)\right)$.

Now, if we look at the argument of the above logarithm more closely, we notice that the body of the sum can be factored in terms of expressions that only depend on $I[i]$. It follows by sum separation that:
\[
\sum_{I} \exp\left(-\sum_{\varphi\in\Phi}\sum_{\substack{i \in \tids(I) \\ I[i]\models grd(i,\varphi)}} w_\varphi \right) = \prod_{i \in \tids(\D)} \sum_{t \in \tuples(R^i)} \exp \sum_{\varphi \in \Phi} \sum_{I[i]\models grd(i,\varphi)} -w_\varphi
\]
From this we have that:
\begin{align*}
& \log \left( \prod_{i \in \tids(\D)} \sum_{t \in \tuples(R^i)} \exp \sum_{\varphi \in \Phi} \sum_{I[i]\models grd(i,\varphi)} -w_\varphi \right) \\
&= \sum_{i \in \tids(\D)}\log \left( \sum_{t \in \tuples(R^i)} \exp \sum_{\varphi \in \Phi} \sum_{I[i]\models grd(i,\varphi)} -w_\varphi\right)
\end{align*}

Based on our discussion in Section~\ref{sec:learningsetup}, we have that for Gibbs parfactor/update PUD models with unary constraints one can assume without loss of generality that the set of tuple identifiers in $\tids(\D)$ is exactly those present in the training examples $(I_{j},J_{j})_{j=1}^n$. Given all the above we can rewrite the negative log-likelihood as:
\begin{align*}
l(\Xi = \mathbf{c},\Theta = \mathbf{d}; (I_j,J_j)_{j=1}^n) = &-\sum_{j=1}^n \sum_{f \in F}w_f f(I_j[0]) \\
&+\sum_{j=1}^n
\sum_{\varphi\in\Phi}
w_\varphi|V(\varphi,I_j[0])|\\
&+ \sum_{j=1}^n \log \left( \sum_{t \in \tuples(R^i)} \exp \sum_{f \in F}w_ff(t) \right)\\
&+ \sum_{j=1}^n \log \left( \sum_{t \in \tuples(R^i)} \exp \sum_{\varphi \in \Phi} \sum_{I[i]\models grd(i,\varphi)} -w_\varphi\right)\\
&- \sum_{j=1}^n \sum_{g \in G}w_g
  g(I_j[0],J_j[0])\\
&+\sum_{j=1}^n\log \left( \sum_{t^\prime \in \tuples(R^i)}\exp\left(\sum_{g \in
    G}w_g g(I_{j}[0],t^\prime)\right) \right)
\end{align*}

Eventually we have that:
\begin{align*}
l^\prime(\Xi = \mathbf{c},\Theta = \mathbf{d}; I_j[0],J_j[0]) = &-\sum_{f \in F}w_f f(I_j[0]) + \sum_{\varphi\in\Phi}
w_\varphi|V(\varphi,I_j[0])| \\
&-\sum_{g \in G}w_g
  g(I_j[0],J_j[0])\\
&+\log \left( \sum_{t \in \tuples(R^i)} \exp \sum_{f \in F}w_ff(t) \right)\\
&+\log \left( \sum_{t \in \tuples(R)} \exp \sum_{\varphi \in \Phi} \sum_{I[i]\models grd(i,\varphi)} -w_\varphi\right)\\
&+\log \left( \sum_{t^\prime \in \tuples(R)}\exp\left(\sum_{g \in
    G}w_g g(I_{j}[0],t^\prime)\right) \right)
\end{align*}
where $R$ denotes the relation symbol $R \in \scs$ associated with the tuple identifier in $I_{j}[0]$.

Function $l^\prime(\Xi = \mathbf{c},\Theta = \mathbf{d}; I_j[0],J_j[0]$ is convex over $\Xi$ and $\Theta$ as it corresponds to the sum of LogSumExp functions with affine functions.
\end{proof}

Second we show that:
\def\proplinear{
The gradient of  $l^\prime(\mathbf{c},\mathbf{d}; I_j[0],J_j[0])$ with respect to $\Xi$ and $\Theta$ can be evaluated in time linear to $\tuples(R)$ where $R$ is the relation corresponding to the tuple identifier associated with example $(I_{j},J_{j})$.
}

\begin{proposition}\label{prop:linear}
\proplinear
\end{proposition}
\begin{proof}
To compute the gradient of $l^\prime(\Xi = \mathbf{c},\Theta = \mathbf{d}; I_j[0],J_j[0])$ we need to compute the partial derivatives with respect to each parameter $w_f$, $w_\varphi$ and $w_g$. It is trivial to see that to compute these partial derivates for the first three terms of $l^\prime(\Xi = \mathbf{c},\Theta = \mathbf{d}; I_j[0],J_j[0])$ one must evaluate each feature function $f$, $g$ and $V(\varphi, \cdot)$. All of these functions are assumed to be efficiently computable. It is also trivial to see that to compute the partial derivate of each LogSumExp term of $l^\prime(\Xi = \mathbf{c},\Theta = \mathbf{d}; I_j[0],J_j[0])$ one needs to iterate over all tuples $t \in \tuples(R)$ as the expression inside the logarithm appear in the denominator of the each partial derivative. Hence, the time required to compute the gradient for $l^\prime(\Xi = \mathbf{c},\Theta = \mathbf{d}; I_j[0],J_j[0])$ is $O(\tuples(R))$.
\end{proof}

\begin{reptheorem}{\ref{thm:convex_indep}}
\thmconvexindep
\end{reptheorem}

\begin{proof}
We have that $l(\mathbf{c},\mathbf{d}; (I_j,J_j)_{j=1}^n) = \sum_{j=1}^n l^\prime(\mathbf{c},\mathbf{d}; I_j[0],J_j[0])$. To compute the overall gradient of $l$ we need to compute the gradient of each function $l^\prime$. Hence, the overall complexity is $O(n\cdot \max_{R \in \scs}|\tuples(R)|)$. 
\end{proof}

\subsection{Proof of Theorem~\ref{thm:converge}}
\label{sec:convergence}

Before we present the proof for this theorem we discuss some notation that we use for convenience and we also discuss the tools used to prove asymptotic normality. 

First, we switch to matrix notation to denote the Gibbs parametric intention and realization models $\I_{\Xi}$ and $\R_{\Theta}$. It is a simple exercise to show that that parametric models introduced in Section~\ref{sec:learningsetup} can be written as $\I_{\Xi}(I) = \frac{1}{Z_{\I}} \exp (\Xi^T \cdot u_{\I}(I)$ and $\R_{\Theta}(J|I) = \frac{1}{Z_{\R}(I)} \exp (\Theta^T \cdot u_{\R}(I,J)$ where $u_{\I}(\cdot)$ is a vector function that corresponds to the features characterizing $\I_{\Xi}$ and $u_{\R}(\cdot)$ is a vector function that corresponds to the features characterizing $\R_{\Theta}$. These models correspond to the standard {\em exponential family}~\cite{Wainwright:2008:GME:1498840.1498841}.

Vector functions $u_{\I}(\cdot)$ and $u_{\R}(\cdot)$ can also be represented as matrices $u_{\I}: \R^{|\Omega| \times d_{\I}}$ and $u_{\R}: \R^{|\Omega| \times d^{\R}}$ where $\Omega$ is the sample space of our PUD (e.g., if we had a single relation $R$ that would be $\tuples(R)$), $d_{\I}$ is the number of features describing $\I_{\Xi}$, and $d_{\R}$ is the number of features describing $\R_{\Theta}$.

Second, we introduce the notion of {\em Fisher information} of the available training data~\cite{kullback1997information}. The Fisher information determines the amount of information that observed database instances carry about the unknown parameters $\Xi$ and $\Theta$. Intuitively, Fisher information can be interpreted as a measure of how quickly the distribution density will change when we slightly change a parameter in $\theta$ near the optimal $\theta^*$. 

Next, we define the Fisher information of a Gibbs parfactor/update PUD:

\begin{definition}
\label{defnFischerInformation}
The \emph{intention's Fisher information} of a Markov parametric PUD is:
\[
  \mathcal{I}(\Xi)
  =
  \mathbf{Cov}_{I \sim \I_{\Xi}}\left[ u_{\I}(I) \right].
\]
\indent \indent Similarly, the \emph{realizer's Fisher information} of a parametric PUD is:
\[
  \mathcal{I}(\Theta)
  =
  \mathbf{E}_{I \sim \I_{\Xi}}\left[
    \mathbf{Cov}_{J \sim \R_{\Theta}(\cdot | I)}\left[ u_{\R}(J|I) \right]
  \right].
\]
\end{definition}

For general parameter learning, the Fisher information matrices can be
{\em singular}, i.e., our observations carry no information about the
parameters in some direction. Two conditions that lead to singular
Fisher information matrices are: (1) the parameters of
our PUD model are {\em redundant} (e.g., we can have the same formula
listed twice with different weights) or (2) there is a parameter in
our PUD model that has no effect on the distribution (e.g., we have a
parameter associated with a formula that always evaluates to
true). Notice that both cases described above correspond to
misspecified parametric PUD models. 

We now show that for any Gibbs parfactor/update PUD model the Fisher information matrices are positive definite, thus, not singular.
Given that the Fisher information matrices are not singular, and thus invertible, we show the {\em asymptotic normality} of the MLE estimates $\Xi$ and $\Theta$ for tuple independent MLD/update PUDs. 

\begin{lemma}
\label{lemmaFischerPositive}
For any PUD, if matrix $u_{\I}$ is always full-rank, and
similarly for $u_{\R}$, and all parameters $(\Xi, \Theta)$
are finite, then
\[
  \mathcal{I}(\Xi) \succ 0 \text{ and } \mathcal{I}(\Theta) \succ 0.
\]
That is that $\mathcal{I})(\Xi)$ and $\mathcal{I}(\Theta)$ are positive definite and thus invertible.
\end{lemma}

\begin{proof}
  For the intention model, the Fisher information is the covariance of $u_\I(I)$.
  The only way this matrix could be singular is if there is some unit vector $\phi$ such that
  \[
    \mathbf{Var}_{I \sim \I_{\Xi}}[\phi^T u_\I(I)] = 0.
  \]
  This, in turn, will only happen if $\phi^T u_\I(I)$ is constant across all $I$ on which $\I_{\Xi}$ is supported.
  Since $\Xi$ is finite, $\I_{\Xi}$ is supported everywhere on $\Omega$.
  This means that if we define $\bar u_\I$ as
  \[
    \bar u_{\I}(I) = u_\I(I) - \phi \phi^T u_\I(I),
  \]
  then $\bar u_{\I}(I)$ will also be a sufficient statistics function for the same PUD.
  But, $\bar u_{\I}(I)$ is rank deficient, because $\phi^T \bar u_{\I} = 0$.
  But this cannot happen, since we supposed that any sufficient statistics matrix would be full rank.
  Therefore, the Fisher information is positive definite, which is what we wanted to prove.
\end{proof}

Given that the Fisher information matrices are invertible we can now proceed to show asymptotic normality. Before we proceed with our theorem, recall that since we consider Gibbs parfactor/update PUDs with unary constraints each example $(I_{j}, J_{j})$ in the collection of training examples $(I_{j}, J_{j})_{j=1}^n)$ corresponds to an independent single-tuple database example. In the theorem below, we use quantities $\I_{1}(\Xi^*)$ and $\I_{1}(\Theta^*)$ to denote the Fisher information of a Gibbs parfactor/update PUD model over a single-tuple database. We have for asymptotic normality:

\begin{reptheorem}{\ref{thm:converge}}
\thmconverge
\end{reptheorem}

\begin{proof}
  We will prove this by the standard proof technique that is used to prove asymptotic normality.
  The gradient of the negative log-likelihood of an exponential family model with features $u$ and parameters $\theta$ is:
  \begin{align*}
    \nabla f(\theta)
    &=
    -\nabla_{\theta}
    \log\left(
      \frac{\exp(\theta^T u(x))}{\sum_{y \in \Omega} \exp(\theta^T u(y))}
    \right) \\
    &=
    -u(x)
    +
    \frac{\sum_{y \in \Omega} \exp(\theta^T u(y)) u(y)}{\sum_{y \in \Omega} \exp(\theta^T u(y))} =
    -u(x)
    +
    \mathbf{E}_{y \sim \pi_{\theta}}\left[ u(y) \right]
  \end{align*}
  and the Hessian (the matrix that corresponds to the second-order partial derivatives) is
  \begin{align*}
    \nabla^2 f(\theta)
    &=
    \frac{\sum_{y \in \Omega} \exp(\theta^T u(y)) u(y) u(y)^T}{\sum_{y \in \Omega} \exp(\theta^T u(y))}
    \,\, -
    \\ &\hspace{2em}    
    \frac{\sum_{y \in \Omega} \exp(\theta^T u(y)) u(y)}{\sum_{y \in \Omega} \exp(\theta^T u(y))}
    \cdot
    \frac{\sum_{y \in \Omega} \exp(\theta^T u(y)) u(y)^T}{\sum_{y \in \Omega} \exp(\theta^T u(y))} \\
    &=
    \mathbf{E}_{x \sim \pi_{\theta}}\left[ u(x) u(x)^T \right]
    -
    \mathbf{E}_{x \sim \pi_{\theta}}\left[ u(x) \right]
    \mathbf{E}_{x \sim \pi_{\theta}}\left[ u(x) \right]^T =
    \mathbf{Cov}_{x \sim \pi_{\theta}}\left[ u(x) \right].
  \end{align*}
  In the limit, we have that the gradient at the true parameter values is, for training examples $x_1, \ldots, x_n$,
  \begin{align*}
    \nabla f(\theta)
    &=
    \mathbf{E}_{y \sim \pi_{\theta}}\left[ u(y) \right]
    -
    \frac{1}{n}
    \sum_{i=1}^n u(x_i).
  \end{align*}
  By a Taylor expansion, we expand the negative log-likelihood about the true parameters $\theta^\star$ and obtain that:
  \[
    \nabla f(\theta) = H (\theta - \theta^*).
  \]
  From the above and Lemma~\ref{lemmaFischerPositive}, which states that the Fisher information is positive definite and thus invertible, it follows that
  \begin{align*}  
    \theta - \theta^*
    &=
    \I(\theta)^{-1} \left(
      \mathbf{E}_{y \sim \pi_{\theta}}\left[ u(y) \right]
      -
      \frac{1}{n}
      \sum_{i=1}^n u(x_i).
    \right)\\
    &=
    \bigg( \mathbf{Cov}_{x \sim \pi_{\theta}}\big[ u(x) \big] \bigg)^{-1} \bigg(
      \mathbf{E}_{y \sim \pi_{\theta}}\bigg[ u(y) \bigg]
      -
      \frac{1}{n}
      \sum_{i=1}^n u(x_i).
    \bigg)
  \end{align*}
  where $\I(\theta) = \mathbf{Cov}_{x \sim \pi_{\theta}}\left[ u(x) \right]$ is the Fisher information of the model for $x$.
  This has expected value $0$, and covariance
  \[
    \mathbf{Cov}\left[\theta - \theta^* \right]
    =
    \frac{1}{n}
    \mathbf{Cov}_{x \sim \pi_{\theta}}\left[ u(x) \right]^{-1}.
  \]
  If we plug in our Gibbs parfactor/update PUD model to the above we have for the intention model:
  \begin{equation*} 
  	\mathbf{Cov}[\Xi - \Xi^\star] = \frac{1}{n}\mathbf{Cov}_{I\sim\I_{\Xi}}\left[ u_\I(I) \right]^{-1}
  \end{equation*}
  Therefore:
  \begin{equation*}
  	\mathbf{Cov}[\Xi - \Xi^\star] = \frac{1}{n}\mathcal{I}_1(\Xi^\star)^{-1} \text{~as~} n \rightarrow \infty.
  \end{equation*}
  For a conditional exponential family distribution, we have:
  \begin{align*}
    \pi(X | I)
    =
    \frac{
      \exp(\theta^T u(X|I))
    }{
      \sum_Y \exp(\theta^T u(Y|I))
    }.
  \end{align*}
  The gradient of the negative log-likelihood is
  \begin{align*}
    \nabla_{\theta} -\log \pi(X | I)
    &=
    -u(X|I)
    +
    \frac{
      \sum_Y \exp(\theta^T u(Y|I)) u(Y|I)
    }{
      \sum_Y \exp(\theta^T u(Y|I))
    } \\& =
    -u(X|I) + \mathbf{E}_{Y \sim \pi(\cdot|I)}\left[ u(Y|I) \right].
  \end{align*}
  The Hessian is
  \begin{align*}
    \nabla_{\theta}^2 -\log \pi(X | I)
    &=
    \frac{
      \sum_Y \exp(\theta^T u(Y|I)) u(Y|I) u(Y|I)^T
    }{
      \sum_Y \exp(\theta^T u(Y|I))
    }\\
    &-
    \left(
      \frac{
        \sum_Y \exp(\theta^T u(Y|I)) u(Y|I)
      }{
        \sum_Y \exp(\theta^T u(Y|I))
      }
    \right)
    \left(
      \frac{
        \sum_Y \exp(\theta^T u(Y|I)) u(Y|I)
      }{
        \sum_Y \exp(\theta^T u(Y|I))
      }
    \right)^T \\
    &=
    \mathbf{Cov}_{Y \sim \pi(\cdot|I)}\left[ u(Y|I) \right].
  \end{align*}
  In expectation over $I$, this will be
  \begin{align*}
    \nabla^2 f(\theta)
    &=
    \mathbf{E}_I\left[ \mathbf{Cov}_{Y \sim \pi(\cdot|I)}\left[ u(Y|I) \right] \right].
  \end{align*}
  This leaves us with a typical gradient of using samples
  \begin{align*}
    \nabla f(\theta)
    &=
    \frac{1}{n} \sum_{k=1}^n \left(
      \mathbf{E}_{Y \sim \pi(\cdot|I_k)}\left[ u(Y|I_k) \right]
      -
      u(X_k|I_k)
    \right).
  \end{align*}
  This will have expected value $0$, and covariance
  \begin{align*}
    \mathbf{Cov}\left[ \nabla f(\theta) \right]
    &=
    \mathbf{Cov}\left[ 
      \frac{1}{n} \sum_{k=1}^n \left(
        \mathbf{E}_{Y \sim \pi(\cdot|I_k)}\left[ u(Y|I_k) \right]
        -
        u(X_k|I_k)
      \right)
    \right] \\
    &=
    \frac{1}{n}
    \mathbf{Cov}\left[
      \mathbf{E}_{Y \sim \pi(\cdot|I_1)}\left[ u(Y|I_1) \right]
      -
      u(X_1|I_1)
    \right] \\
    &=
    \frac{1}{n}
    \mathbf{E}_{I_1} \left[ \mathbf{Cov}_{X_1} \left[
      \mathbf{E}_{Y \sim \pi(\cdot|I_1)}\left[ u(Y|I_1) \right]
      -
      u(X_1|I_1)
    \right] \right]\\ &=
    \frac{1}{n}
    \mathbf{E}_{I_1} \left[ \mathbf{Cov}_{X_1} \left[ u(X_1|I_1) \right] \right].
  \end{align*}
  So in this setting,
  \[
    \mathbf{Cov}\left[ \theta - \theta^* \right]
    =
    \frac{1}{n}
    \mathbf{E}_z\left[ \mathbf{Cov}_{x \sim \pi_{\theta}}\left[ u(x|z) \right]^{-1} \right].
  \]
  
If we plug in our Gibbs parfactor/update PUD model to the above we have for the realizer model:
  \begin{equation*} 
  	\mathbf{Cov}[\Theta - \Theta^\star] = \frac{1}{n}\mathbf{E}_{I \sim \I_{\Xi}}\mathbf{Cov}_{J\sim\R_{\Theta}(\cdot|I)}\left[ u_\R(J|I) \right]^{-1}
  \end{equation*}
  Therefore:
  \begin{equation*}
  	\mathbf{Cov}[\Theta - \Theta^\star] = \frac{1}{n}\mathcal{I}_1(\Theta^\star)^{-1} \text{~as~} n \rightarrow \infty.
  \end{equation*}
This concludes the proof.
\end{proof}

\subsection{Proof of Theorem~\ref{thm:convexnll}}
For convenience we use the notation introduced at the beginning of Section~\ref{sec:convergence}. We continue with our proof.

\begin{reptheorem}{\ref{thm:convexnll}}
\convexnll
\end{reptheorem}

\begin{proof}
Based on the discussion in Section~\ref{sec:learningsetup}, we have that $(J_{j})_{j=1}^n = J^*$. Given this, the negative log-likelihood with respect to $\Xi$ is:
  \begin{equation*}
    l(\Xi)
    =
    -\log\left(\sum_{I\in \Omega}\I_{\Xi}(I)\cdot\R_{\Theta}(J^\star|I)\right)=
    -\log\left(
      \sum_{I \in \Omega}
      \frac{
        \exp(\Xi^T u_\I(I))
      }{
        \sum_{J \in \Omega} \exp(\Theta^T u_\I(J))
      }
      \cdot \R_{\Theta}(J^\star|I)
    \right).
  \end{equation*}
  The gradient of this is
  \begin{align*}
    \nabla f(\Xi)
    &=
    \frac{
      \sum_{J \in \Omega} \exp(\Xi^T u_\I(J)) \cdot u_\I(J)
    }{
      \sum_{J \in \Omega} \exp(\Xi^T u_\I(J) )
    }
    -
    \frac{
      \sum_{I \in \Omega} \exp(\Xi^T u_\I(I)) \cdot \R_{\Theta}(J^\star|I) u_\I(I)
    }{
      \sum_{I \in \Omega} \exp(\Xi^T u_\I(I)) \cdot \R_{\Theta}(J^\star|I)
    }
  \end{align*}
  and the Hessian is
  \begin{align*}
    &\nabla^2 f(\Xi)
    =
    \frac{
      \sum_{J \in \Omega} \exp(\Xi^T u_\I(J)) \cdot u_\I(J) \cdot u_\I(J)
    }{
      \sum_{J \in \Omega} \exp(\Xi^T u_\I(J))
    }\\
    &-
    \left(
      \frac{
        \sum_{J \in \Omega} \exp(\Xi^T u_\I(J)) \cdot u_\I(J)
      }{
        \sum_{J \in \Omega} \exp(\Xi^T u_\I(J))
      }
    \right)
    \left(
      \frac{
        \sum_{J \in \Omega} \exp(\Xi^T u_\I(J)) \cdot u_\I(J)
      }{
        \sum_{J \in \Omega} \exp(\Xi^T u_\I(J))
      }
    \right)^T\\
    &-
    \frac{
      \sum_{I \in \Omega} \exp(\Xi^T u_\I(I)) \cdot \R_{\Theta}(J^\star|I) \cdot u_\I(I) \cdot u_\I(I)
    }{
      \sum_{I \in \Omega} \exp(\Xi^T u_\I(I)) \cdot \R_{\Theta}(J^\star|I)
    }\\
    &+
    \bigg(
      \frac{
        \sum_{I \in \Omega} \exp(\Xi^T u_\I(I)) \cdot \R_{\Theta}(J^\star|I) \cdot u_\I(I) \cdot u_\I(I)
      }{
        \sum_{I \in \Omega} \exp(\Xi^T u_\I(I)) \cdot \R_{\Theta}(J^\star|I)
      }
    \bigg)\\
    &\quad\,\,    
    \bigg(
      \frac{
        \sum_{I \in \Omega} \exp(\Xi^T u_\I(I)) \cdot \R_{\Theta}(J^\star|I) \cdot u_\I(I) \cdot u_\I(I)
      }{
        \sum_{I \in \Omega} \exp(\Xi^T u_\I(I)) \cdot \R_{\Theta}(J^\star|I)
      }
    \bigg)^T \\
    &=
    \mathbf{Cov}[u_\I(I)] - \mathbf{Cov}[u_\I(I) | J^\star].
  \end{align*}
  By the argument in Lemma~\ref{lemmaFischerPositive}, we know that $\mathbf{Cov}[u_\I(I)] \succ 0$ on all $\Xi $.
  It follows by continuity that there exists a $\delta$ such that on $\Xi$, $\mathbf{Cov}[u_\I(I)] \succ \delta \mathbf{I}$ where $\mathbf{I}$ is the identity matrix.
  On the other hand, if we have a $\R_{\Theta}(J^\star|I)$ such that for each tuple identifier $i \in I$ $I[i] = J^\star[i]$ with probability at least $1-p$, then
  \[
    \mathbf{Cov}[u_\I(I) | J^\star] \preceq p \max_{J \in \Omega} \|u_I(J)\|^2.
  \] 
  It follows that we can choose a $p$ small enough that the Hessian is always positive definite on $\Xi$, which means $f$ is convex.
  This completes the proof.
\end{proof}

 Notice that for the Hessian to be positive definite it must be that $p\cdot \max_{J \in \Omega} \|u_\I(J)\|^2 \prec \delta I$ which in turn means that the value of $p$ depends on the maximum value of features $u_\I$ when computed over $J$. This means that probability $p$ might not be bounded. In Section~\ref{sec:learndirty}, we discuss methods for solving the unsupervised version of PUD learning. These methods require that $p$ is bounded. 
 
Based on the analysis for the above theorem, we show that $p$ is bounded for Gibbs parfactor/update PUD models with unary constraints. We have the following proposition:

\def\bounded{
Given a Gibbs parfactor/update PUD model with unary constraints  
for which the low noise condition holds, then we have for probability $p$ that $p\cdot \|C\|^2 < \delta$
where $C = \max_{R \in \scs, t \in \tuples(R), f \in F, \varphi \in \Phi}(f(t),V(\phi,t))$ and $\delta$ is a constant with $\delta > 0$.
}

\begin{proposition}\label{prop:bounded}
\bounded	
\end{proposition}
\begin{proof}
From the proof of Theorem~\ref{thm:convexnll} we have that when the noise condition holds it must be that:
\[
    p\cdot \max_{J \in \Omega}\|u_{\I}(J)\|^2 \prec \delta \mathbf{I}
\]
It is easy to see that in the case of Gibbs parfactor/update PUD models with unary constraints the maximum value for $\|u_{\I}(J)\|^2$ for any $J$ scales independently of the number of tuple identifiers in $\tids(\D)$ and depends only on the features of distribution $\K$ for the intention model $\I_{\Xi}$.
\end{proof}



\end{document}